\newif\ifsubmission\submissionfalse

\ifsubmission
\documentclass{llncs}
\usepackage[T1]{fontenc}
\else
\documentclass{article}
\usepackage[utf8]{inputenc}
\usepackage{fullpage}
\fi

\RequirePackage{amsmath,amsfonts,amssymb,amsthm}
\RequirePackage{complexity}
\RequirePackage[inline]{enumitem}
\RequirePackage{textcomp}
\RequirePackage{xspace}
\RequirePackage[normalem]{ulem}
\RequirePackage[usenames,dvipsnames]{xcolor}
\RequirePackage{mdframed}
\RequirePackage{newfloat,caption}
\RequirePackage{tikz}
\usetikzlibrary{positioning,arrows}
\RequirePackage[hyperfootnotes=false,hypertexnames=false,colorlinks=true,linkcolor=Maroon,citecolor=Maroon]{hyperref}
\RequirePackage{multicol}
\RequirePackage{comment}
\RequirePackage{algorithm,algpseudocode}
\RequirePackage{thm-restate}
\RequirePackage[capitalize]{cleveref}

\DeclareFloatingEnvironment[placement={!h},name=Protocol]{protocol}

\ifsubmission
\newtheorem{fact}          {Fact}{\itshape}{\bfseries}
\else
\theoremstyle{plain}
\newtheorem{theorem}       {Theorem}           [section]
\newtheorem{inf-theorem}   {Informal Theorem}  [section]

\newtheorem{claim}         {Claim}             [section] 
\newtheorem{fact}          {Fact}              [section]
\fi

\ifsubmission
\else
\theoremstyle{definition}
\newtheorem{definition}         {Definition}     [section]
\fi

\theoremstyle{remark}
\newenvironment{proofof}[1]{\begin{proof}[\textit{Proof of #1}]}{\end{proof}}

\newenvironment{proofsketchof}[1]{\begin{proof}[\textit{Proof Sketch of #1}]}{\end{proof}}

\crefname{claim}{Claim}{Claims}
\crefname{fact}{Fact}{Facts}

\def\ShowAuthNotes{0}
\ifnum\ShowAuthNotes=1
\newcommand{\snote}[1]{[{\footnotesize  \textcolor{red}{\bf Sanjam:} {#1}}]}
\newcommand{\pnote}[1]{[{\footnotesize  \textcolor{cyan}{\bf Prashant:} {#1}}]}
\newcommand{\sanjam}[1]{\textcolor{red}{Sanjam: #1}}
\newcommand{\shafi}[1]{\textcolor{green}{Shafi: #1}}
\else
\newcommand{\snote}[1]{}
\newcommand{\pnote}[1]{}
\newcommand{\sanjam}[1]{}
\newcommand{\shafi}[1]{}
\fi
\newcommand{\remove}[1]{}

\newcommand{\Nat}           {\mathbb{N}}

\DeclareMathOperator{\supp}{Supp}

\DeclareMathOperator*{\Exp}{E}

\DeclareMathOperator*{\sd}{\Delta}
\let\poly\relax
\DeclareMathOperator{\poly}{poly}

\newcommand{\set}[1]        {\left\{ #1 \right\}}
\newcommand{\abs}[1]        {\left| #1\right|}

\newcommand{\paren}[1]      {\left( #1 \right)}
\newcommand{\pr}[1]         {\Pr\left[ #1 \right]}

\newcommand{\expec}[2]      {\Exp_{#1}\left[ #2 \right]}

\newcommand{\bset}          {\set{0,1}}

\newcommand{\algfont}[1]    {\mathsf{#1}}

\renewcommand{\sim}         {\algfont{Sim}}

\def\ddefloop#1{\ifx\ddefloop#1\else\ddef{#1}\expandafter\ddefloop\fi}

\def\ddef#1{\expandafter\def\csname bb#1\endcsname{\ensuremath{\mathbb{#1}}}}
\ddefloop ABCDEFGHIJKLMNOPQRSTUVWXYZ\ddefloop

\def\ddef#1{\expandafter\def\csname sf#1\endcsname{\ensuremath{\mathsf{#1}}}}
\ddefloop ABCDEFGHIJKLMNOPQRSTUVWXYZabcdefghijklmnopqrstuvwxyz\ddefloop

\def\ddef#1{\expandafter\def\csname c#1\endcsname{\ensuremath{\mathcal{#1}}}}
\ddefloop ABCDEFGHIJKLMNOPQRSTUVWXYZ\ddefloop

\def\ddef#1{\expandafter\def\csname v#1\endcsname{\ensuremath{\boldsymbol{#1}}}}
\ddefloop ABCDEFGHIJKLMNOPQRSTUVWXYZabcdefghijklmnopqrstuvwxyz\ddefloop

\def\ddef#1{\expandafter\def\csname alg#1\endcsname{\ensuremath{\algfont{#1}}}}
\ddefloop ABCDEFGHIJKLMNOPQRSTUVWXYZabcdefghijklmnopqrstuvwxyz\ddefloop

\def\ddef#1{\expandafter\def\csname v#1\endcsname{\ensuremath{\boldsymbol{\csname #1\endcsname}}}}
\ddefloop {alpha}{beta}{gamma}{delta}{epsilon}{varepsilon}{zeta}{eta}{theta}{vartheta}{iota}{kappa}{lambda}{mu}{nu}{xi}{pi}{varpi}{rho}{varrho}{sigma}{varsigma}{tau}{upsilon}{phi}{varphi}{chi}{psi}{omega}{Gamma}{Delta}{Theta}{Lambda}{Xi}{Pi}{Sigma}{varSigma}{Upsilon}{Phi}{Psi}{Omega}{ell}\ddefloop

\newcommand{\ra}            {\rightarrow}

\newcommand{\eps}           {\varepsilon}

\newcommand{\secp}          {\lambda}

\newcommand{\idfont}[1]     {\mathnormal{#1}}

\newcommand{\sid}           {\idfont{sId}}
\newcommand{\pid}           {\idfont{pId}}

\newcommand{\sdist}[2]         {\sd\paren{#1 ; #2}}

\newcommand{\state}         {\mathnormal{state}}
\newcommand{\view}          {\mathnormal{view}}
\newcommand{\alg}[1]        {\mathsf{#1}}
\newcommand{\ins}[1]        {\mathsf{#1}}
\ifsubmission
\renewcommand{\inst}          {\ins{instruction}}
\else
\newcommand{\inst}          {\ins{instruction}}
\fi

\algnewcommand{\OR}{\textbf{or}}
\algnewcommand{\AND}{\textbf{and}}

\newcommand{\entfont}[1]    {\mathcal{#1}}
\newcommand{\X}             {\entfont{X}}
\newcommand{\Z}             {\entfont{Z}}
\newcommand{\Y}             {\entfont{Y}}
\newcommand{\diffp}         {\entfont{D}}
\newcommand{\ml}            {\entfont{M}}
\renewcommand{\W}           {\entfont{W}}
\newcommand{\histind}       {\entfont{H}}

\newcommand{\exec}{\textsc{EXEC}}
\newcommand{\cexec}{\textsc{AEXEC}}
\newcommand{\er}{\textsc{EXEC}^{\X,\pi,\pi_D}_{\Z,\Y}}
\newcommand{\ei}{\textsc{EXEC}^{\X,\pi,\pi_D}_{\Z,\Y_0}}
\newcommand{\erphi}{\textsc{EXEC}^{\X,\pi,\pi_D}_{\Z,\Y,\Phi}}
\newcommand{\eiphi}{\textsc{EXEC}^{\X,\pi,\pi_D}_{\Z,\Y_0,\Phi}}
\newcommand{\erpi}{\textsc{EXEC}^{\X,\pi,\pi_D}_{\Z,\Y,(\pi,\pi_D)}}
\newcommand{\eipi}{\textsc{EXEC}^{\X,\pi,\pi_D}_{\Z,\Y_0,(\pi,\pi_D)}}

\newcommand{\dict}          {\alg{Dict}}

\newenvironment{datacol}{\vspace{0.5em}\noindent\dotfill\vspace{0.5em}\newline\small}{\vspace{-0.5em}\mbox{}\noindent\dotfill\vspace{0.5em}}

\title{Formalizing Data Deletion in the Context of the Right to be Forgotten}
\ifsubmission
\author{}
\institute{}
\else
\author{Sanjam Garg\thanks{EECS, UC Berkeley. Email: \texttt{\{sanjamg,prashvas\}@berkeley.edu}. Supported in part from AFOSR Award FA9550-19-1-0200, AFOSR YIP Award, NSF CNS Award 1936826, DARPA and SPAWAR under contract N66001-15-C-4065, a Hellman Award and research grants by the Okawa Foundation, Visa Inc., and Center for Long-Term Cybersecurity (CLTC, UC Berkeley). The views expressed are those of the authors and do not reflect the official policy or position of the funding agencies.} \and Shafi Goldwasser\thanks{Simons Institute for the Theory of Computing, UC Berkeley. Email: \texttt{shafi@theory.csail.mit.edu}. Supported in part by the C. Lester Hogan Chair in EECS, UC Berkeley, and Fintech@CSAIL.} \\ \and Prashant Nalini Vasudevan\footnotemark[1]}
\fi
\date{}

\begin{document}
\maketitle

\vspace{-2em}

\begin{abstract}
    The right of an individual to request the deletion of their personal data by an entity that might be storing it -- referred to as \emph{the right to be forgotten} --  has been explicitly recognized, legislated, and exercised in several jurisdictions across the world, including the European Union, Argentina, and California. However, much of the discussion surrounding this right offers only an intuitive notion of what it means for it to be fulfilled -- of what it means for such personal data to be deleted.
    
    In this work, we provide a formal definitional framework for the right to be forgotten using tools and paradigms from cryptography. In particular, we provide a precise definition of what could be (or should be) expected from an entity that collects individuals' data when a request is made of it to delete some of this data. Our framework captures several, though not all, relevant aspects of typical systems involved in data processing. While it cannot be viewed as expressing the statements of current laws (especially since these are rather vague in this respect), our work offers technically precise definitions that represent possibilities for what the law could reasonably expect, and alternatives for what future versions of the law could explicitly require.
    
    Finally, with the goal of demonstrating the applicability of our framework and definitions, we consider various natural and simple scenarios where the right to be forgotten comes up. For each of these scenarios, we highlight the pitfalls that arise even in genuine attempts at implementing systems offering deletion guarantees, and also describe technological solutions that provably satisfy our definitions. These solutions bring together techniques built by various communities.
\end{abstract}



\section{Introduction}
\label{sec:intro}

Everything we do in our lives leaves (or will soon leave) a digital trace, which can be analyzed. Recent advances in capturing and analyzing big data help us improve traffic congestion, accurately predict human behavior and needs in various situations, and much more. However, this mass collection of data can be used against people as well. Simple examples of this would be to charge individuals higher auto insurance premiums or decline mortgages and jobs based on an individual's profile as presented by the collected data. In the worst case, this wealth of information could be used by totalitarian governments to persecute their citizens years after the data was collected. In such ways, vast collection of personal data has the potential to present a serious infringement to personal liberty. Individuals could perpetually or periodically face stigmatization as a consequence of a specific past action, even one that has already been adequately penalized. This, in turn, threatens democracy as a whole, as it can force individuals to self-censor personal opinions and actions for fear of later retaliation.


One alternative for individuals wanting to keep personal information secret is to simply stay offline, or at least keep such information hidden from entities that are likely to collect it. Yet, this is not always desirable or possible. These individuals might want to share such information with others over an internet-based platform, or obtain a service based on their personal information, such as personalized movie recommendations based on previous movie watching history, or simply driving directions to their destination based on where they want to go. In such cases, it is reasonable to expect that an individual might later change their mind about having this data available to the service provider they sent it to. In order to provide useful functionality while keeping in mind the aforementioned perils of perennial persistence of data, an individual's ability to withdraw previously shared personal information is very important. For example, one might want to request deletion of all personal data contained in one's Facebook account. 

However, in many cases, an individual's desire to request deletion of their private data may be in conflict with a data collector's\footnote{Throughout this paper, we refer to any entity collecting individuals' data as a ``data collector'', and often refer such indivisuals whose data is collected as ``users''.} interests. In particular, the data collector may want to preserve the data because of financial incentives or simply because fulfilling these requests is expensive. It would seem that, in most cases, the data collector has nothing to gain from fulfilling such requests.

Thus, it seems imperative to have in place legal or regulatory means to grant individuals control over what information about them is possessed by different entities, how it is used, and, in particular, provide individuals the rights to request deletion of any (or all) of their personal data. And indeed, the legitimacy of this desire to request deletion of personal data is being increasingly widely discussed, codified in law, and put into practice (in various forms) in, for instance, the European Union (EU)~\cite{GDPR}, Argentina~\cite{carter2013argentina}, and California~\cite{CCPA}. The following are illustrative examples:

\begin{itemize}
    \item The General Data Protection Regulation (GDPR)~\cite{GDPR}, adopted in 2016, is a regulation in the EU aimed at protecting the data and privacy of individuals in the EU. Article 6 of the GDPR lists conditions under which an entity may lawfully process personal data. The first of these conditions is when ``the data subject has given consent to the processing of his or her personal data for one or more specific purposes''. And Article 7 states that, ``The data subject shall have the right to withdraw his or her consent at any time''. Further, Article 17 states that, ``The data subject shall have the right to obtain from the controller the erasure of personal data concerning him or her without undue delay and the controller shall have the obligation to erase personal data without undue delay'' under certain conditions listed there.

    \item The California Consumer Privacy Act (CCPA), passed in 2018, is a law with similar purposes protecting residents of California. Section 1798.105 of the CCPA states, ``A consumer shall have the right to request that a business delete any personal information about the consumer which the business has collected from the consumer'', and that ``A business that receives a verifiable request from a consumer \dots\ shall delete the consumer's personal information from its records.''
\end{itemize}

Thus, if a data collector (that operates within the jurisdictions of these laws) wishes to process its consumers' data based on their consent, and wishes to do so lawfully, it would also need to have in place a mechanism to stop using any of its consumers' data. Only then can it guarantee the consumers' \emph{right to be forgotten} as the above laws require. However, it is not straightforward to nail down precisely what this means and involves.

\paragraph{Defining Deletion: More that Meets the Eye.} Our understanding of what it means to forget a user's data or honor a user deletion request is rather rudimentary, and consequently, the law does not precisely define what it means to delete something. Further, this lack of understanding is reflected in certain inconsistencies between the law and what would naturally seem desirable. For example, Article 7 of the GDPR, while describing the right of the data subject to withdraw consent for processing of personal data, also states, ``the withdrawal of consent shall not affect the lawfulness of processing based on consent before its withdrawal.'' This seems to suggest that it is reasonable to preserve the result of processing performed on user data even if the data itself is requested to be deleted. However, processed versions of user data may encode all or most of the original data, perhaps even inadvertently. For instance, it is known that certain machine learning models end up memorizing the data they were trained on~\cite{SRS17,VBE18}.

Thus, capturing the intuitive notion of what it means to truly delete something turns out be quite tricky. In our quest to do so, we ask the following question:

\begin{center}
    \emph{How does an honest data collector know whether it is in compliance with the right to be forgotten?}
\end{center}

Here, by \emph{honest} we mean a data collector that does in fact intend to guarantee its users' right to be forgotten in the intuitive sense -- it wishes to truly forget all personal data it has about them. Our question is about how it can tell whether the algorithms and mechanisms it has in place to handle deletion requests are in fact working correctly.


\paragraph{Honest data-collectors.} In this work, we focus on the simple case where the data-collector is assumed to be honest. 
In other words, we are only interested in the data-collectors that aim to faithfully honor all legitimate deletion requests. Thus, we have no adversaries in our setting. This deviates from many cryptographic applications where an adversary typically attempts to deviate from honest execution. Note that even in the case of semi-honest advesaries in multiparty computation, the adversary attempts to learn more than what it is supposed to learn while following protocol specification. In our case, we expect the data-collector to itself follow the prescribed procedures, including deleting any stored information that it is directed to delete.

With the above view, we do not attempt to develop methods by which a data collector could prove to a user that it did indeed delete the user's data. As a remark, we note here that this is in fact impossible in general, as a malicious data collector could always make additional secret copies of user data.\footnote{Certifying deletion could be possible in specific settings though, such as under assumptions on the amount of storage available to the data collector~\cite{PT10,DKW11,KK14}, or in the presence of quantum computers and data~\cite{CW19,BI19}.}
Finally, we note that even for this case of law-abiding data-collectors, the problem of defining what it means to delete data correctly is relevant. The goal of our definitions is to provide such data-collectors guidance in designing systems that handle data deletion, and a mechanism to check that any existing systems are designed correctly and are following the law (or some reasonable interpretation of it).


\paragraph{When is it okay to delete?} Another challenge a data-collector faces in handling deletion requests is in establishing whether a particular deletion request should be honored. Indeed, in some cases a data collector may be legally required to preserve certain information to satisfy legal or archival needs, e.g. a data collector may be required to preserve some payment information that is evidence in a case in trial. This raises the very interesting question of how to determine whether a particular deletion request should indeed be honored, or even what factors should be taken into consideration while making this decision. 
However, this is not the focus of this work. Instead, we are only interested in cases where the data-collector does intend (or has already decided) to honor a received deletion request, after having somehow found it legitimate. In such cases, we aim to specify the requirements this places on the data-collector.

\paragraph{Our Contributions.}
In this work, we provide the first precise general notions of what is required of an honest data-collector trying to faithfully honor deletion requests. We say that a data-collector is \emph{deletion-compliant} if it satisfies our requirements. Our notions are intended to capture the intuitive expectations a user may have when issuing deletion requests. Furthermore, it seems to satisfy the requirements demanded, at least intuitively, by the GDPR and CCPA. However, we note that our definition should not be seen as being equivalent to the relevant parts of these laws -- for one, the laws themselves are somewhat vague about what exactly they require in this respect, and also there are certain aspects of data-processing systems that are not captured by our framework (see \cref{sec:discussion} for a discussion). Instead, our work offers technically precise definitions for data deletion that represent possibilities for interpretations of what the law could reasonably expect, and alternatives for what future versions of the law could explicitly require.

Next, armed with these notions of deletion-compliance, we consider various natural scenarios where the right to be forgotten comes up. For each of these scenarios, we highlight the pitfalls that arise even in genuine attempts at writing laws or honest efforts in implementing systems with these considerations. Our definitions provide guidance towards avoiding these pitfalls by, for one, making them explicit as violations of the definitions. In particular, for each of the considered scenarios, we describe technological solutions that provably satisfy our definitions. These solutions bring together techniques built by various communities.

\subsection{Our Notions}
In this subsection, we explain our notions of deletion-compliance at a high level, building them up incrementally so as to give deeper insights. The formal definitions are in terms of building blocks from the UC framework~\cite{FOCS:Canetti01}, and details are provided in \cref{sec:explanation}.

\paragraph{The starting challenge.} We start with the observation that a deletion request almost always involves much more than the process of just erasing something from memory. In fact, this issue comes up even in the most seemingly benign deletion requests. For example, consider the very simple case where a user requests deletion of one of her files stored with a data-collector. In this setting, even if the server was to erase the file from its memory, it may be the case that not all information about it has been deleted. For example, if the files are stored contiguously in memory, it might be possible to recover the size of the file that was deleted. Furthermore, if the files of a user are kept on contiguous parts of the memory, it might be possible to pin-point the owner of the deleted file as well, or in most cases at least be able to tell that there was a file that was deleted. 

\paragraph{Our approach: leave no trace.} In order to account for the aforementioned issues, we take the \emph{leave-no-trace} approach to deletion in our definitions. In particular, a central idea of our definition is that execution of the deletion request should leave the data collector and the rest of the system in a state that is equivalent (or at least very similar) to one it would have been in if the data that is being deleted was never provided to the data-collector in the first place. 

The requirement of leave-no-trace places several constraints on the data-collector. First, and obviously, the data that is requested to be deleted should no longer persist in the memory of the data-collector after the request is processed. Second, as alluded to earlier, the data-collector must also remove the dependencies that other data could have on the data that is requested for deletion. Or at least, the data-collector should erase the other stored information which depends on this data. We note that we diverge from the GDPR in this sense, as it only requires deletion of data rather than what may have been derived from it via processing. Third, less obvious but clearly necessary demands are placed on the data-collector in terms of what it is allowed to do with the data it collects. In particular, the data-collector cannot reveal any data it collects to any external entity. This is because sharing of user data by the data-collector to external entities precludes it from honoring future deletion requests for the shared data. More specifically, on sharing user data with an external entity, the data-collector loses its the ability to ensure that the data can be deleted from everywhere where it is responsible for the data being present or known. That is, if this data were never shared with the data collector, then it would not have found its way to the external entity, and thus in order for the system to be returned to such a state after a deletion request, the collector should not reveal this data to the entity.

A more concrete consequence of the third requirement above is that the data-collector cannot share or sell user data to third parties. Looking ahead, in some settings this sharing or selling of user data is functionally beneficial and legally permitted as long as the collector takes care to inform the recipients of such data of any deletion requests. For instance, Article 17 of the GDPR says, ``Where the controller has made the personal data public and is obliged \dots\ to erase the personal data, the controller \dots\ shall take reasonable steps, including technical measures, to inform controllers which are processing the personal data that the data subject has requested the erasure by such controllers of any links to, or copy or replication of, those personal data.'' We later see (in \cref{sec:cond-compliance}) how our definition can be modified to handle such cases and extended to cover data collectors that share data with external entities but make reasonable efforts to honor and forward deletion requests.


\paragraph{The basic structure of the definition.} In light of the above discussion, the basic form of the definition can be phrased as follows. Consider a user $\Y$ that shares certain data with a data-collector and later requests for the shared data to be deleted. We refer to this execution as a \emph{real world} execution. In addition to this user, the data-collector might interact with other third parties. In this case, we are interested in the memory state of the data-collector post-deletion and the communication between the data-collector and the third parties. Next, we define the \emph{ideal world} execution, which is same as the real world execution except that the  user $\Y$  does not share anything with the data-collector and does not issue any deletion requests. Here again we are interested in the memory state of the data-collector and the communication between the data-collector and the third parties. More specifically, we require that the joint distribution of memory state of the data-collector and the communication between the data-collector and the third parties in the two worlds is identically distributed (or is at least very close). Further, this property needs to hold not just for a specific user, but hold for every user that \emph{might} interact with the data-collector as part of its routine operation where it is interacting with any number of other users and processing their data and deletion requests as well. Note that the data-collector does not a priori know when and for what data it will receive deletion requests.

\paragraph{A more formal notion.} Hereon, we refer to the data-collector as $\X$, and the deletion requester as $\Y$. In addition to these two entities, we model all other parties in the system using $\Z$, which we also refer to as the environment. Thus, in the real execution, the data-collector $\X$ interacts arbitrarily with the environment $\Z$. Furthermore, in addition to interactions with $\Z$, $\X$ at some point receives some data from $\Y$ which $\Y$ at a later point also requests to be deleted. In contrast, in the ideal execution, $\Y$ is replaced by a silent $\Y_0$ that does not communicate with $\X$ at all. In both of these executions, the environment $\Z$ represent both the rest of the users in the system under consideration, as well as an adversarial entity that possibly instructs $\Y$ on what to do and when. Finally, our definition requires that the state of $\X$ and the view of $\Z$ in the real execution and the ideal execution are similar. Thus, our definition requires that the deletion essentially has the same effect as if the deleted data was never sent to $\X$ to begin with. The two executions are illustrated in \cref{fig:worlds}

\begin{figure}[h]
    \centering
    \includegraphics[width=0.85\textwidth]{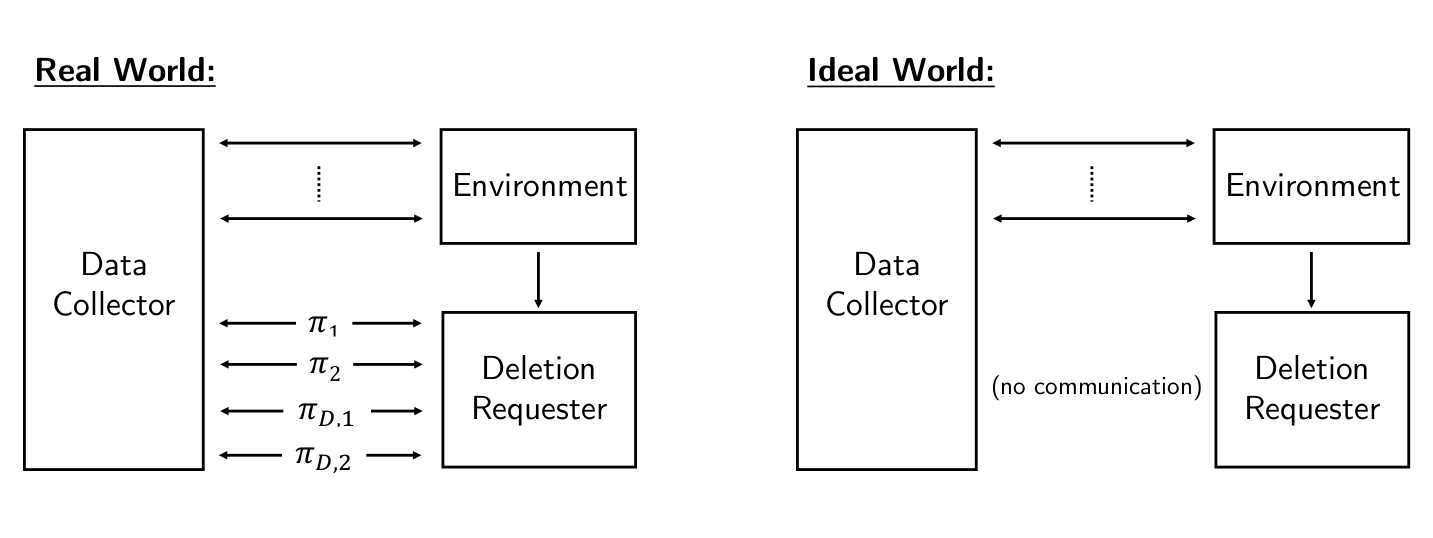}
    \caption{The real and ideal world executions. In the real world, the deletion-requester talks to the data collector, but not in the ideal world. In the real world, $\pi_1$ and $\pi_2$ are interactions that contain data that is asked to be deleted by the deletion-requester through the interactions $\pi_{D,1}$ and $\pi_{D,2}$, respectively.}
    \label{fig:worlds}
\end{figure}

While $\Y$ above is represented as a single user sending some data and a corresponding deletion request, we can use the same framework for a more general modeling. In particular, $\Y$ can be used to model just the part of a user that contains the data to be deleted, or of multiple users, all of whom want some or all of their data to be deleted. 

\paragraph{Dependencies in data.} While the above definition makes intuitive sense, certain user behaviors can introduce dependencies that make it impossible for the data-collector to track and thus delete properly. Consider a data-collector that assigns a pseudonym to each user, which is computed as the output of a pseudo-random permutation $P$ (with the seed kept secret by the data-collector) on the user identity. Imagine a user who registers in the system with his real identity $id$ and is assigned the pseudonym $pd$. Next, the user re-registers a fresh account using $pd$ as his identity. Finally, the user requests deletion of the first account which used his real identity $id$. In this case, even after the data-collector deletes the requested account entirely, information about the real identity $id$ is still preserved in its memory, i.e. $P^{-1}(pd) = id$. Thus, the actions of the user can make it impossible to keep track of and properly delete user data. In our definition, we resolve this problem by limiting the communication between $\Y$ and $\Z$. We do not allow $\Y$ to send any messages to the environment $\Z$, and require that $\Y$ ask for all (and only) the data it sent to be deleted. This implicitly means that the data that is requested to be deleted cannot influence other information that is stored with the data-collector, unless that is also explicitly deleted by the user.

\paragraph{Requirement that the data-collector be diligent.} Our definitions of deletion-compliance place explicit requirements on the data collector only when a deletion request is received. Nonetheless, these explicit requirements implicitly require the data-collector to organize (or keep track of the collected data) in a way that ensures that deletion requests can be properly handled. For example, our definitions implicitly require the data-collector to keep track of how it is using each user's data. In fact, this book-keeping is essential for deletion-compliance. After all, how can a data-collector delete a user's data if it does not even know where that particular user's data is stored? Thus, a data-collector that follows these implicit book-keeping requirements can be viewed as being \emph{diligent}. Furthermore, it would be hard (if not impossible) for a data-collector to be deletion-compliant if it is not diligent.

As we discuss later, our definition also implies a requirement on the data-collector to have in place authentication mechanisms that ensure that it is sharing information only with the legitimate parties, and that only the user who submitted a piece of data can ask for it to be deleted.



\paragraph{Composition Properties.} Finally, we also show, roughly, that under an assumption that different users operate independently of each other, a data collector that is deletion-compliant under our definition for a deletion request from a single user is also deletion-compliant for requests from (polynomially) many users (or polynomially many independent messages from a single user). This makes our definition easier to use in the analysis of certain data collectors, as demonstrated in our examples in \cref{sec:scenarios}. 

\subsection{Lessons from our Definitions}

Our formalization of the notion of data deletion enables us to design and analyze mechanisms that handle data obtained from others and process deletion requests, as demonstrated in \cref{sec:scenarios}. This process of designing systems that satisfy our definition has brought to light a number of properties such a mechanism needs to have in order to be deletion-compliant that may be seen as general principles in this respect.

To start with, satisfying our definition even while providing very simple functionalities requires a non-trivial authentication mechanism that uses randomness generated by the server. Otherwise many simple attacks can be staged that lead to observable differences based on whether some specific data was stored and deleted or never stored. The easier case to observe is when, as part of its functionality, the data collector provides a way for users to retrieve data stored with it. In this case, clearly if there is no good authentication mechanism, then one user can look at another user's data and be able to remember it even after the latter user has asked the collector to delete it. More broadly, our definition implicitly requires the data collector to provide certain privacy guarantees -- that one user's data is not revealed to others.

 But even if such an interface is not provided by the collector, one user may store data in another user's name, and then if the latter user ever asks for its data to be deleted, this stored data will also be deleted, and looking at the memory of the collector after the fact would indicate that such a request was indeed received. If whatever authentication mechanism the collector employs does not use any randomness from the collector's side, such an attack may be performed by any adversary that knows the initial state (say the user name and the password) of the user it targets.

Another requirement that our definition places on data collectors is that they handle metadata carefully. For instance, care has to be taken to use implementations of data structures that do not inadvertently preserve information about deleted data in their metadata. This follows from our definition as it talks about the state of the memory, and not just the contents of the data structure. Such requirements may be satisfied, for instance, by the use of ``history-independent'' implementations of data structures~\cite{Mic97,NT01}, which have these properties.

Further, this kind of history-independence in other domains can also be used to provide other functionalities while satisfying our definition. For instance, recent work~\cite{CY15,GGVZ19,ECSSL19,GAS19,Schelter20,BCCJTZLP20,BSZ20} has investigated the question of data deletion in machine learning models, and this can be used to construct a data collector that learns such a model based on data given to it, and can later delete some of this data not just from its database, but also from the model itself.

Finally, we observe that privacy concepts, such as differential privacy~\cite{DMNS06}, can sometimes be used to satisfy deletion requirements without requiring any additional action from the data collector at all. Very roughly, a differentially private algorithm guarantees that the distribution of its output does not change by much if a small part of its input is changed. We show that if a data collector runs a differentially private algorithm on data that it is given, and is later asked to delete some of the data, it need not worry about updating the output of the algorithm that it may have stored (as long as not too much data is asked to be deleted). Following the guarantee of differential privacy, whether the deleted data was used or not in the input to this algorithm essentially does not matter.

\subsection{Related Work}
\label{sec:related}

Cryptographic treatment of legal terms and concepts has been undertaken in the past. Prominent examples are the work of Cohen and Nissim~\cite{CN19} that formalizes and studies the notion of singling-out that is specified in the GDPR as a means to violate privacy in certain settings, and the work of Nissim et al~\cite{nissim2017bridging} that models the privacy requirements of FERPA using a game-based definition.

Recently, the notion of data deletion in machine learning models has been studied by various groups~\cite{CY15,GGVZ19,ECSSL19,GAS19,Schelter20,BCCJTZLP20,BSZ20}. Closest to our work is the paper of Ginart et al~\cite{GGVZ19}, which gives a definition for what it means to retract some training data from a learned model, and shows efficient procedures to do so in certain settings like $k$-means clustering. We discuss the crucial differences between our definitions and theirs in terms of scope and modelling in \cref{sec:discussion}.

There has been considerable past work on notions of privacy like differential privacy~\cite{DMNS06} that are related to our study, but very different in their considerations. Roughly, in differential privacy, the concern is to protect the privacy of each piece of data in a database -- it asks that the output of an algorithm running on this database is roughly the same whether or not any particular piece of data is present. We, in our notion of deletion-compliance, ask for something quite different -- unless any piece of data is requested to be deleted, the state of the data collector could depend arbitrarily on it; only \emph{after} this deletion request is processed by the collector do the requirements of our definition come in. In this manner, while differential privacy could serve as a means to satisfy our definition, our setting and considerations in general are quite different from those there. For similar reasons, our definitions are able to require bounds on statistical distance without precluding all utility (and in some cases even perfect deletion-compliance is possible), whereas differential privacy has to work with a different notion of distance between distributions (see \cite[Section 1.6]{Vadhan17} for a discussion).

While ours is the first formal definition of data deletion in a general setting, there has been considerable work on studying this question in specific contexts, and in engineering systems that attempt to satisfy intuitive notions of data deletion, with some of it being specifically intended to support the right to be forgotten. We refer the reader to the comprehensive review article by Politou et al~\cite{PAP18} for relevant references and discussion of such work.

%
%


\section{Our Framework and Definitions}
\label{sec:defs}

In this section we describe our framework for describing and analyzing data collectors, and our definitions for what it means for a data collector to be \emph{deletion-compliance}. Our modeling uses building blocks that were developed for the Universal Composability (UC) framework of Canetti~\cite{FOCS:Canetti01}. First, we present the formal description of this framework and our definitions. Explanations of the framework and definitions, and how we intend for them to be used are given in \cref{sec:explanation}. In \cref{sec:discussion}, we discuss the various choices made in our modelling and the implicit assumptions and restrictions involved. In \cref{sec:cond-compliance}, we present a weakening of our definition that covers data collectors that share data with external entities, and in \cref{sec:props} we demonstrate some composition properties that our definition has.

\paragraph{The Model of Execution.}

Looking ahead, our approach towards defining deletion-compliance of a data collector will be to execute it and have it interact with certain other parties, and at the end of the execution ask for certain properties of what it stores and its communication with these parties. Following~\cite{GolMicRac89,Goldreich01,FOCS:Canetti01}, both the data collector and these other parties in our framework are modelled as \emph{Interactive Turing Machines} (ITMs), which represent the program to be run within each party. Our definition of an ITM is very similar to the one in \cite{CCL15}, but adapted for our purposes.

\begin{definition}[Interactive Turing Machine]
    \label{def:itm}
    An Interactive Turing Machine (ITM) is a (possibly randomized) Turing Machine $M$ with the following tapes:
    \begin {enumerate*} [label=(\itshape\roman*\upshape)]
      \item a read-only \emph{identifier tape};
      \item a read-only \emph{input tape};
      \item a write-only \emph{output tape};
      \item a read-write \emph{work tape};
      \item a single-read-only \emph{incoming tape};
      \item a single-write-only \emph{outgoing tape};
      \item a read-only \emph{randomness tape}; and
      \item a read-only \emph{control tape}. 
    \end{enumerate*}
    
    
    The \emph{state} of an ITM $M$ at any given point in its execution, denoted by $\state_M$, consists of the content of its work tape at that point. 
    Its \emph{view}, denoted by $\view_M$, consists of the contents of its input, output, incoming, outgoing, randomness, and control tapes at that point. 
\end{definition}

The execution of the system consists of several instances of such ITMs running and reading and writing on their own and each others' tapes, and sometimes instances of ITMs being created anew, according to the rules described in this subsection. We distinguish between ITMs (which represent static objects, or programs) and \emph{instances of ITMs}, or ITIs, that represent instantiations of that ITM. Specifically, an ITI is an ITM along with an identifer that distinguishes it from other ITIs in the same system. This identifier is written on the ITI's identifier tape at the point when the ITI is created, and its semantics will be described in more detail later.

In addition to having the above access to its own tapes, each ITI, in certain cases, could also have access to read from or write on certain tapes of other ITI. The first such case is when an ITI $M$ \emph{controls} another ITI $M'$. $M$ is said to control the ITIs whose identifiers are written on its control tape, and for each ITI $M'$ on this tape, $M$ can read $M'$'s output tape and write on its input tape. This list is updated whenever, in the course of the execution of the system, a new ITI is created under the control of $M$.

The second case where ITIs have access to each others' tapes is when they are engaged in a \emph{protocol}. A protocol is described by a set of ITMs that are allowed to write on each other's incoming tapes. Further, any ``message'' that any ITM writes on any other ITM's incoming tape is also written on its own outgoing tape. As with ITMs, a protocol is just a description of the ITMs involved in it and their prescribed actions and interactions; and an \emph{instance} of a protocol, also referred to as a \emph{session}, consists of ITIs interacting with each other (where indeed some of the ITIs may deviate from the prescribed behavior). Each such session has a unique session identifier ($\sid$), and within each session each participating ITI is identified by a unique party identifier ($\pid$). The identifier corresponding to an ITI participating in a session of a protocol with session identifier $\sid$ and party identifier $\pid$ is the unique tuple $(\sid,\pid)$.

There will be small number of special ITIs in our system, as defined below, whose identifiers are assigned differently from the above. Unless otherwise specified, all ITMs in our system are probabilistic polynomial time (PPT) -- an ITM $M$ is PPT if there exists a constant $c > 0$ such that, at any point during its run, the overall number of steps taken by $M$ is at most $n^c$, where $n$ is the overall number of bits written on the \emph{input tape} of $M$ during its execution.


\paragraph{The Data Collector.} We require the behavior of the data collector and its interactions with other parties to be specified by a tuple $(\X,\pi,\pi_D)$, where $\X$ specifies the algorithm run by the data collector, and $\pi,\pi_D$ are protocols by means of which the data collector interacts with other entities. Here, $\pi$ could be an arbitrary protocol (in the simplest case, a single message followed by local processing), and $\pi_D$ is the corresponding \emph{deletion} protocol -- namely, a protocol to undo/reverse a previous execution of the protocol $\pi$.

For simplicity, in this work, we restrict to protocol $\pi,\pi_D$ to the natural case of the two-party setting.\footnote{However, our model naturally generalizes to protocols with more parties.} Specifically, each instance of the protocol $\pi$ that is executed has specifications for a server-side ITM and a client-side ITM. The data collector will be represented in our system by a special ITI that we will also refer to as $\X$. When another ITI in the system, call it $\W$ for now, wishes to interact with $\X$, it does by initiating an instance (or session) of one of the protocols $\pi$ or $\pi_D$. This initiation creates a pair of ITIs -- the client and the server of this session -- where $\W$ controls the client ITI and $\X$ the server ITI. $\W$ and $\X$ then interact by means of writing to and reading from the input and output tapes of these ITIs that they control. Further details are to be found below.


The only assumption we will place on the syntax of these protocols is the following interface between $\pi$ and $\pi_D$. We require that at the end of any particular execution of $\pi$, a \emph{deletion token} is defined that is a function solely of the $\sid$ of the execution and its transcript, and that $\pi$ should specify how this token is computed. The intended interpretation is that a request to delete this instance of $\pi$ consists of an instance of $\pi_D$ where the client-side ITI is given this deletion token as input. As we will see later, this assumption does not lose much generality in applications.

\paragraph{Recipe for Describing Deletion-Compliance.} Analogous to how security is defined in the UC framework, we define \emph{deletion-compliance} in three steps as follows. First, we define a \emph{real execution} where certain other entities interact with the data collector ITI $\X$ by means of instances the protocols $\pi$ and $\pi_D$. This is similar to the description of the ``real world'' in the UC framework. In this setting, we identify certain deletion requests (that is, executions of $\pi_D$) that are of special interest for us -- namely, the requests that we will be requiring to be satisfied . Next, we define an \emph{ideal execution}, where the instances of $\pi$ that are asked to be deleted by these identified deletion requests are never executed in the first place. The ``ideal execution'' in our setting is different from the ``ideal world'' in the UC framework in the sense that we do not have an ``ideal functionality''. Finally, we say that $(\X,\pi,\pi_D)$ is \emph{deletion-compliant} if the two execution process are essentially the same in certain respects. Below, we explain the model of the \emph{real} execution, the \emph{ideal} execution, and the notion of deletion-compliance.

\paragraph{Real Execution.} The real execution involves the data collector ITI $\X$, and two other special ITIs: the \emph{environment} $\Z$ and the \emph{deletion requester} $\Y$. By intention, $\Y$ represents the part of the system whose deletion requests we focus on and will eventually ask to be respected by $\X$, and $\Z$ corresponds to the the rest of the world -- the (possibly adversarial) environment that interacts with $\X$. Both of these interact with $\X$ via instances of $\pi$ and $\pi_D$, with $\X$ controlling the server-side of these instances and $\Z$ or $\Y$ the client-side.

The environment $\Z$, which is taken to be adversarial, is allowed to use arbitrary ITMs (ones that may deviate from the protocol) as the client-side ITIs of any instances of $\pi$ or $\pi_D$ it initiates. The deletion-requester $\Y$, on the other hand, is the party we are notionally providing the guarantees for, and is required to use honest ITIs of the ITMs prescribed by $\pi$ and $\pi_D$ in the instances it initiates, though, unless otherwise specified, it may provide them with any inputs as long as they are of the format required by the protocol.\footnote{Note that it is essential that $\Y$ follow the honest protocol specifications to ensure that the deletion requests are successful.} In addition, we require that any instance of $\pi_D$ run by $\Y$ is for an instance of $\pi$ already initiated by $\Y$.\footnote{This corresponds to providing guarantees only for entities that do not (maliciously or otherwise) ask for others' data to be deleted.} Finally, in our modeling, while $\Z$ can send arbitrary messages to $\Y$ (thereby influencing its executions), we do not allow any communication from $\Y$ back to $\Z$. This is crucial for ensuring that the $\X$ does not get any ``to be deleted'' information from other sources. 

At any point, there is at most one ITI in the system that is \emph{activated}, meaning that it is running and can reading from or writing to any tapes that it has access to. Each ITI, while it is activated, has access to a number of tapes that it can write to and read from. Over the course of the execution, various ITIs are activated and deactivated following rules described below. When an ITI is activated, it picks up execution from the point in its ``code'' where it was last deactivated.

Now we provide a formal description of the real execution. We assume that all parties have a computational/statistical security parameter $\secp \in \mathbb{N}$ that is written on their input tape as $1^\secp$ the first time they are activated.\footnote{We remark that this is done merely for convenience and is not essential for the model to make sense. In particular, in the perfect security case, no security parameter is needed.} The execution consists of a sequence of \emph{activations}, where in each activation a single participant (either $\Z$, $\Y$, $\X$ or some ITM) is activated, and runs until it writes on the incoming tape of another (at most \emph{one} other) machine, or on its own output tape. Once this write happens, the writing participant is deactivated (its execution is paused), and another party is activated next --- namely, the one on who incoming tape the message was written; or alternatively, if the message was written to the output tape then the party controlling the writing ITI is activated. If no message is written to the incoming tape (and its own output tape) of any party, then $\Z$ is activated. The real execution proceeds in two phases: (i) the alive phase, and (ii) the termination phase. 

\paragraph{Alive Phase:} This phase starts with an activation of the environment $\Z$, and $\Z$ is again activated if any other ITI halts without writing on a tape. The various ITIs run according to their code, and are allowed to act as follows: 
\begin{itemize} 
  \item The environment $\Z$ when active is allowed to read the tapes it has access to, run, and perform any of the following actions:
    \begin{itemize}
      \item Write an arbitrary message on the incoming tape of $\Y$.
      \item Write on the input tape of any ITI that it controls (from protocol instances initiated in the past).
      \item Initiate a new protocol instance of $\pi$ or $\pi_D$ with $\X$, whereupon the required ITIs are created and $\Z$ is given control of the client-side ITI of the instance and may write on its input tape. At the same time, $\X$ is given control of the corresponding server-side ITI that is created.
      \item Pass on activation to $\X$ or $\Y$.
      \item Declare the end of the Alive Phase, upon which the execution moves to the Terminate Phase. This also happens if $\Z$ halts.
    \end{itemize}
    
  \item The deletion-requester $\Y$ on activation can read the tapes it has access to, run, and perform any of the following actions:
    \begin{itemize}
      \item Write on the input tape of any ITI that it controls.
      \item Initiate a new instance of $\pi$ or $\pi_D$ with $\X$, and write on the input tape of the created client-side ITI.
    \end{itemize}

  \item The data collector $\X$ on activation can read the tapes it has access to, run, and write on the input tape of any ITI that it controls.
    
  \item Any other ITI that is activated is allowed to read any of the tapes that it has access to, and write to either the incoming tape of another ITI in the protocol instance it is a part of, or on its own output tape.
\end{itemize}

\paragraph{Terminate Phase:} In this phase, the various ITIs are allowed the same actions as in the Alive phase. The activation in this phase proceeds as follows:
\begin{enumerate}
  \item First, each client-side ITI for $\pi$ that was initiated by $\Y$ in the Alive phase is sequentially activated enough times until each one of them halts.
  \item For any instance of $\pi$ for which a client-side ITI was initiated by $\Y$ and which was executed to completion, an instance of $\pi_D$ is initiated with input the deletion token for that instance of $\pi$ (except if such an instance of $\pi_D$ was already initiated).
  \item Each client-side ITI for instances of $\pi_D$ that were initiated by $\Y$ in the Alive phase or in the previous step is sequentially activated enough times until each one of them halts.
\end{enumerate}

We denote by $\er(\secp)$ the tuple $(\state_\X,\view_\X,\state_\Z,\view_\Z)$ resulting at the end of above-described real execution with security parameter $\secp$.

\paragraph{Ideal Execution.} Denote by $\Y_0$ the special $\Y$ that is completely silent -- whenever it is activated, it simply halts. In particular, it does not initiate any ITIs and does not write on the incoming tape of any other machine. A real execution using such a $\Y_0$ as the deletion-requester is called an ideal execution. We denote by $\er(\secp)$ the tuple $(\state_\X,\view_\X,\state_\Z,\view_\Z)$ resulting at the end of an ideal execution with data collector $\X$ and environment $\Z$, and with security parameter $\secp$.

\paragraph{} We are now ready to present our definition for the deletion-compliance of data collectors, which is as follows.

\begin{definition}[Statistical Deletion-Compliance]
    \label{def:compliance}
    Given a data-collector $(\X,\pi,\pi_D)$, an environment $\Z$ and a deletion-requester $\Y$, let $(\state_\X^{R,\secp},\view_\Z^{R,\secp})$ denote the corresponding parts of the real execution $\er(\secp)$, and $(\state_\X^{I,\secp},\view_\Z^{I,\secp})$ the corresponding parts of the ideal execution $\ei(\secp)$. We say that $(\X, \pi,\pi_D)$ is \emph{statistically deletion-compliant} if, for any PPT environment $\Z$, any PPT deletion-requester $\Y$, and for all unbounded distinguishers $D$, there is a negligible function $\eps$ such that for all $\secp \in \Nat$:
    \begin{align*}
      \left| \Pr[D(\state_\X^{R,\secp},\view_\Z^{R,\secp})=1] - \Pr[D(\state_\X^{I,\secp},\view_\Z^{I,\secp})=1]\right| \leq \eps(\secp)
    \end{align*}
\end{definition}

In other words, the statistical distance between these two distributions above is at most $\eps(\secp)$. If $D$ above is required to be computationally bounded (allowed to run only in PPT time in $\secp$), then we get the weaker notion of \emph{computational deletion-compliance}. Analogously, if $\eps(\secp)$ is required to be $0$, then we get the stronger notion of \emph{perfect deletion-compliance}.


\subsection{Explanation of the Definition}
\label{sec:explanation}

As indicated earlier, the central idea our definition is built around is that the processing of a deletion request should leave the data collector and the rest of the system in a state that is similar to one it would have been in if the data that was deleted was never given to the collector in the first place. This ensures that there is no trace left of deleted data, even in metadata maintained by some of the entities, etc..

The first question that arises here is which parts of the system to ask this of. It is clear that the deleted data should no longer persist in the memory of the data collector. A less obvious but clearly necessary demand is that the data collector also not reveal this data to any user other than the one it belongs to. Otherwise, unless whomever this data is revealed to provides certain guarantees for its later deletion, the data collector loses the ability to really delete this data from locations it reached due to actions of the data collector itself, which is clearly undesirable.\footnote{Of course, if the entity this data is revealed to does provide some guarantees for later deletion, then we may reasonably expect the data collector to provide deletion guarantees even while revealing data to this entity. In \cref{sec:cond-compliance}, we present a weaker definition of deletion-compliance that captures this.}

Once so much is recognized, the basic form of the definition is clear from a cryptographic standpoint. We fix any user, let the user send the collector some data and then request for it to be deleted, and look at the state of the collector at this point together with its communication with the rest of the system so far. We also look at the same in a world where this user did not send this data at all. And we ask that these are distributed similarly. We then note that this property needs to hold not just when the collector is interacting solely with this user, but is doing so as part of its routine operation where it is interacting with any number of other users and processing their data and deletion requests as well. 

\paragraph{The UC Framework.} In order to make this definition formal, we first need to model all entities in a formal framework that allows us to clearly talk about the ``state'' or the essential memory of the entities, while also being expressive enough to capture all, or at least most, data collectors. We chose the UC framework for this purpose as it satisfies both of these properties and is also simple enough to describe clearly and succinctly. In this framework, the programs that run are represented by Interactive Turing Machines, and communication is modelled as one machine writing on another's tape. The state of an entity is then captured by the contents of the work tape of the machine representing it, and its view by whatever was written on its tapes by other machines. This framework does impose certain restrictions on the kind of executions that it captures, though, and this is discussed later, in \cref{sec:discussion}.

\paragraph{Protocols and Interaction.} Another choice of formality motivated by its usefulness in our definition is to have all communication with the data collector $\X$ be represented by instances of a protocol $\pi$. It should be noted that the term ``protocol'' here might belie the simplicity of $\pi$, which could just involve the sending of a piece of data by a user of the system to the data collector $\X$. This compartmentalisation of communication into instances of $\pi$ is to let us (and the users) refer directly to specific instances later and request their deletion using instances of the deletion protocol $\pi_D$. As the reference to instances of $\pi$, we use a ``deletion token'' that is computable from the transcript of that instance -- this is precise enough to enable us to refer to specific pieces of data that are asked to be deleted, and loose enough to capture many natural systems that might be implemented in reality for this purpose.

\paragraph{The Deletion-Requester $\Y$ and the Environment $\Z$.} The role of the user in the above rudimentary description is played by the deletion-requester $\Y$ in our framework. In the ``real'' execution, $\Y$ interacts with the data collector $\X$ over some instances of $\pi$, and then asks for all information contained in these instances to be deleted. In the ``ideal'' execution, $\Y$ is replaced by a silent $\Y_0$ that does not communicate with $\X$ at all. And both of these happen in the presence of an environment $\Z$ that interacts arbitrarily with $\X$ (through instances of $\pi$ and $\pi_D$) -- this $\Z$ is supposed to represent both the rest of the users in the system that $\X$ interacts with, as well as an adversarial entity that, in a sense, attempts to catch $\X$ if it is not handling deletions properly. By asking that the state of $\X$ and the view of $\Z$ in both these executions be similar, we are asking that the deletion essentially have the same effect on the world as the data never being sent.

It is to be noted that while $\Y$ here is represented as a single entity, it does not necessarily represent just a single ``user'' of the system or an entire or single source of data. It could represent just a part of a user that contains the data to be deleted, or represent multiple users, all of whom want their data to be deleted. In other words, if a data collector $\X$ is deletion-compliant under our definition, and at some point in time has processed a certain set of deletion requests, then as long as the execution of the entire world at this point can be separated into $\Z$ and $\Y$ that follow our rules of execution, the deletion-compliance of $\X$ promises that all data that was sent to $\X$ from $\Y$ will disappear from the rest of the world.

\paragraph{Using the Definition.} Our framework and definition may be used for two purposes: (i) to guide the design of data collectors $\X$ that are originally described within our framework (along with protocols $\pi$ and $\pi_D$) and wish to handle deletion requests well, and (ii) to analyse the guarantees provided by existing systems that were not designed with our framework in mind and which handle data deletion requests.

In order to use \cref{def:compliance} to analyze the deletion-compliance of pre-existing systems, the first step is to rewrite the algorithm of the data collector to fit within our framework. This involves defining the protocols $\pi$ and $\pi_D$ representing the communication between ``users'' in the system and the data collector. This part of the process involves some subjectivity, and care has to be taken to not lose crucial but non-obvious parts of the data collector, such as metadata and memory allocation procedures, in this process. The examples of some simple systems presented in \cref{sec:scenarios} illustrate this process )though they do not talk about modelling lower-level implementation details). Once the data collector and the protocols are described in our framework, the rest of the work in seeing whether they satisfy our definition of deletion-compliance is well-defined.

\subsection{Discussion}
\label{sec:discussion}

A number of choices were made in the modelling and the definition above, the reasons for some of which are not immediately apparent. Below, we go through a few of these and discuss their place in our framework and definition.

\paragraph{Modelling Interactions.} The first such choice is to include in the model the entire communication process between the data collector and its users rather than look just at what goes on internally in the data collector. For comparison, a natural and simpler definition of data deletion would be to consider a data collector that has a database, and maintains the result of some computation on this database. It then receives requests to delete specific rows in the database, and it is required to modify both the database and the processed information that it maintains so as to make it look like the deleted row was never present. The definition of data deletion in machine learning by Ginart et al~\cite{GGVZ19}, for instance, is of this form.

The first and primary reason for this choice is that the intended scope of our definitions is larger than just the part of the data collector that maintains the data. We intend to analyze the behavior of the data collector as a whole, including the memory used to implement the collector's algorithm and the mechanisms in place for interpreting and processing its interactions with external agents. For instance, as we discuss in \cref{sec:scenarios}, it turns out that any data collector that wishes to provide reasonable guarantees to users deleting their data needs to have in place a non-trivial authentication mechanism. This requirement follows easily from the requirements of our definition, but would not be apparent if only the part of the collector that directly manages the data is considered.

The second reason is that while the simpler kind of definition works well when the intention is to apply it to collectors that do indeed have such a static database that is given to them, it fails to capture crucial issues that arise in a more dynamic setting. Our inclusion of the interactions between parties in our definition enables us to take into account dependencies among the data in the system, which in turn enables us to keep our demands on the data collector more reasonable. Consider, for example, a user who sends its name to a data collector that responds with a hash of it under some secret hash function. And then the user asks the same collector to store a piece of data that is actually the same hash, but there is no indication given to the collector that this is the case. At some later time, the user asks the collector to delete its name. To a definition that only looks at the internal data storage of the collector, the natural expectation after this deletion request is processed would be that the collector's state should look as though it never learnt the user's name. However, this is an unreasonable demand -- since the collector has no idea that the hash of the name was also given to it, it is not reasonable to expect that it also find the hash (which contains information about the name) and delete it. And indeed, under our definition, the collector is forgiven for not doing so unless the user explicitly asks for the hash also to be deleted. If our modelling had not kept track of the interactions between the collector and the user, we would not have been able to make this relaxation.

\paragraph{Restrictions on $\Y$.} Another conspicuous choice is not allowing the deletion-requester $\Y$ in our framework to send messages to the environment $\Z$. This is, in fact, how we handle cases like the one just described where there are dependencies between the messages that the collector receives that are introduced on the users' side. By requiring that $\Y$ does not send messages to $\Z$ and that all interaction between $\Y$ and $\X$ are asked to be deleted over the course of the execution, we ensure that any data that depends on $\X$'s responses to $\Y$'s messages is also asked to be deleted. This admits the case above where both the name and the hash are requested to be deleted, and requires $\X$ to comply with such a request; but it excludes the case where only the name is asked to be deleted (as then the hash would have to be sent by $\Z$, which has no way of learning it), thus excusing $\X$ for not deleting it.

Also note that this restriction does not lose any generality outside of excluding the above kind of dependency. Take any world in which a user (or users) asks for some of its messages to be deleted, and where the above perverse dependency does not exist between these and messages not being asked to be deleted. Then, there is a pair of environment $\Z$ and deletion-requester $\Y$ that simulates that world exactly, and the deletion-compliance guarantees of $\X$ have the expected implications for such a deletion request. The same is true of the restriction that \emph{all} of the messages sent by $\Y$ have to be requested to be deleted rather than just some of them -- it does not actually lose generality. And also of the fact that $\Y$ is a single party that is asking for deletion rather than a collection -- a set of users asking for deletion can be simulated by just one $\Y$ that does all their work.

\paragraph{The Ideal Deletion-Requester.} An interesting variant of our definition would be one in which the $\Y$ is not replaced by a silent $\Y_0$ in the ideal world, but by another $\Y'$ that sends essentially the same kinds of messages to $\X$, but with different contents. Currently, our definition says that, after a deletion request, the collector does not even remember that it had some data that was deleted. This might be unnecessarily strong for certain applications, and this modification would relax the requirement to saying that it is fine for the collector to remember that it had some data that was deleted, just not what the data was. The modification is not trivial, though, as in general the number and kinds of messages that $\Y$ sends could depend on the contents of its messages and the responses from $\X$, which could change if the contents are changed. Nevertheless, under the assumption that $\Y$ behaves nicely in this sense, such an alternative definition could be stated and would be useful in simple applications.

\paragraph{Choices That Lose Generality.} There are certain assumptions in our modelling that do break from reality. One of these is that all machines running in the system are sequential. Due to this, our definition does not address, for instance, the effects of race conditions in the data collector's implementation. This assumption, however, makes our definition much simpler and easier to work with, while still keeping it meaningful. We leave it as an open question to come up with a reasonable generalization of our definition (or an alternative to it) that accounts for parallel processing.

Another such assumption is that, due to the order of activations and the fact that activation is passed on in the execution by ITIs writing on tapes, we do not give $\Z$ the freedom to interlace its messages freely with those being sent by $\Y$ to $\X$. It could happen, for instance, that $\X$ is implemented poorly and simply fails to function if it does not receive all messages belonging to a particular protocol instance consecutively. This failure is not captured by our definition as is, but this is easily remedied by changing the activation rules in the execution to pass activation back to $\Z$ after each message from (an ITI controlled by) $\Y$ to $\X$ is sent and responded to. We do not do this for the sake of simplicity.

Finally, our modelling of the data collector's algorithm being the entire ITM corresponds to the implicit assumption of reality that the process running this algorithm is the only one running on the system. Or, at least, that the distinguisher between the real and ideal worlds does not get to see how memory for this process is allocated among all the available memory in the system, does not learn about scheduling in the system, etc.. Side-channel attacks involving such information and definitions that provide protection against these would also be interesting for future study, though even more exacting than our definition.


\subsection{Conditional Deletion-Compliance}
\label{sec:cond-compliance}

As noted in earlier sections, any data collector that wishes to be deletion-compliant under \cref{def:compliance} cannot reveal the data that is given to it by a user to any other entity. There are several situations, however, where such an action is desirable and even safe for the purposes of deletion. And rules for how the collector should act when it is in fact revealing data in this way is even specified in some laws -- Article 17 of the GDPR, for instance, says, ``Where the controller has made the personal data public and is obliged \dots to erase the personal data, the controller, taking account of available technology and the cost of implementation, shall take reasonable steps, including technical measures, to inform controllers which are processing the personal data that the data subject has requested the erasure by such controllers of any links to, or copy or replication of, those personal data.''

Consider, for instance, a small company $\X$ that offers storage services using space it has rented from a larger company $\W$. $\X$ merely stores indexing information on its end and stores all of its consumers' data with $\W$, and when a user asks for its data to be deleted, it forwards (an appropriately modified version of) this request to the $\W$. Now, if $\W$ is deletion-compliant and deletes whatever data $\X$ asks it to, it could be possible for $\X$ to act in way that ensures that state of the entire system composed of $\X$ and $\W$ has no information about the deleted data. In other words, conditioned on some deletion-compliance properties of the environment (that includes $\W$ here), it is reasonable to expect deletion guarantees even from collectors that reveal some collected data. In this subsection, we present a definition of \emph{conditional} deletion-compliance that captures this. 

Specifically, we consider the case where the environment $\Z$ itself is deletion-compliant, though in a slightly different sense than \cref{def:compliance}. In order to define this, we consider the deletion-compliance of a data collector $\X$ running its protocols $(\pi,\pi_D)$ in the presence of other interaction going on in the system. So far, in our executions involving $(\X,\pi,\pi_D)$, we essentially required that $\Y$ and $\Z$ only interact with $\X$ by means of the protocols $\pi$ and $\pi_D$. Now we relax this requirement and, in both phases of execution, allow an additional set of protocols $\Phi = \set{\phi_1, \dots}$ that can be initiated by $\X$ to be run between $\X$ and $\Z$ (but not $\Y$) during the execution. We denote  an execution involving $\X$, $\Z$ and $\Y$ under these rules by $\exec^{\X,\pi,\pi_D}_{\Z,\Y,\Phi}$.

Finally, we also consider executions where, additionally, we also let $\X$ write on the incoming tape of $\Y$.\footnote{This weakens the definition of deletion-compliance, as it allows $\X$ to send to $\Y$ anything it wants, since the view or state of $\Y$ is not scrutinized by the requirements of deletion-compliance. And though as a definition of deletion-compliance this is not meaningful on its own, it is a property that, if the environment $\Z$ possesses it, seems necessary and sufficient to allow a data collector $\X$ to safely reveal data to $\Z$ that it may wish to delete later.} We call such an execution an \emph{auxiliary} execution, and denote it by $\cexec^{\X,\pi,\pi_D}_{\Z,\Y,\Phi}$. We define the following notion of auxiliary deletion-compliance that we will be the condition we will place on the environment in our eventual definition of conditional deletion-compliance. 

\begin{definition}[Auxiliary Deletion-Compliance]
    \label{def:compliance-aux}
    Given a data-collector $(\X,\pi,\pi_D)$, an environment $\Z$, a deletion-requester $\Y$, and a set of protocols $\Phi$, let $(\state_\X^{R,\secp},\view_\Z^{R,\secp})$ denote the corresponding parts of the auxiliary execution $\cexec^{\X,\pi,\pi_D}_{\Z,\Y,\Phi}(\secp)$, and $(\state_\X^{I,\secp},\view_\Z^{I,\secp})$ the corresponding parts of the ideal auxiliary execution $\cexec^{\X,\pi,\pi_D}_{\Z,\Y_0,\Phi}(\secp)$. We say that $(\X, \pi,\pi_D)$ is \emph{statistically auxiliary-deletion-compliant in the presence of $\Phi$} if, for any PPT environment $\Z$, any PPT deletion-requester $\Y$, and for all unbounded distinguishers $D$, there is a negligible function $\eps$ such that for all $\secp\in\Nat$:
    \begin{align*}
      \left| \Pr[D(\state_\X^{R,\secp},\view_\Z^{R,\secp})=1] - \Pr[D(\state_\X^{I,\secp},\view_\Z^{I,\secp})=1]\right| \leq \eps(\secp)
    \end{align*}
\end{definition}

Note that we do not ask $\X$ for any guarantees on being able to delete executions of the protocols in $\Phi$. It may be seen that any data collector $(\X,\pi,\pi_D)$ that is deletion-compliant is also auxiliary deletion-compliant in the presence of any $\Phi$, since it never runs any of the protocols in $\Phi$.

We say that a data collector $\X$ is conditionally deletion-compliant if, whenever it is interacting with an environment that is auxiliary-deletion-compliant, it provides meaningful deletion guarantees. 

\begin{definition}[Conditional Deletion-Compliance]
    \label{def:cond-compliance}
    Given a data-collector $(\X,\pi,\pi_D)$, an environment $\Z$, a deletion-requester $\Y$, and a pair of protocols $\Phi = (\phi,\phi_D)$, let $(\state_\X^{R,\secp},\state_\Z^{R,\secp})$ denote the corresponding parts of the real execution $\erphi(\secp)$, and $(\state_\X^{I,\secp},\state_\Z^{I,\secp})$ the corresponding parts of the ideal execution $\eiphi(\secp)$. We say that $(\X, \pi,\pi_D)$ is \emph{conditionally statistically deletion-compliant in the presence of $\Phi$} if, for any PPT environment $\Z$ such that $(\Z,\phi,\phi_D)$ is statistically auxiliary-deletion-compliant in the presence of $(\pi,\pi_D)$, any PPT deletion-requester $\Y$, and for all unbounded distinguishers $D$, there is a negligible function $\eps$ such that for all $\secp\in\Nat$:
    \begin{align*}
      \left| \Pr[D(\state_\X^{R,\secp},\state_\Z^{R,\secp})=1] - \Pr[D(\state_\X^{I,\secp},\state_\Z^{I,\secp})=1]\right| \leq \eps(\secp)
    \end{align*}
\end{definition}

One implication of $\X$ being conditionally deletion-compliant is that if, in some execution, it is found that data that was requested of $\X$ to be deleted is still present in the system in some form, then this is not due to a failure on the part of $\X$, but was because the environment $\Z$ was not auxiliary-deletion-compliant and hence failed to handle deletions correctly. A setup like the one described at the beginning of this subsection is studied as an example of a conditionally deletion-compliant data collector in \cref{sec:hist-ind-2}.


\subsection{Properties of our Definitions}
\label{sec:props}

In this section, we demonstrate a few properties of our definition of deletion-compliance that are meaningful to know on their own and will also make analyses of data collectors we design in later sections simpler. In order to describe them, we first define certain special classes of deletion-requesters. The first is one where we limit the number of protocol instances the deletion-requester $\Y$ is allowed to initiate.

\begin{definition}
    \label{def:ky}
    For any $k\in\Nat$, a deletion-requester $\Y$ is said to be \emph{$k$-representative} if, when interacting with a data collector $\X$ running $(\pi,\pi_D)$, it initiates at most $k$ instances of $\pi$ with $\X$.
\end{definition}

The other is a class of deletion-requesters intended to represent the collected actions of several $1$-representative deletion-requesters operating independently of each other. In other terms, the following represents, say, a collection of users that interact with a data collector by sending it a single message each, and further never interact with each other. This is a natural circumstance that arises in several situations of interest, such as when people respond to a survey or submit their medical records to a hospital, for example. Hence, even deletion-compliance guarantees that hold only in the presence of such deletion-requesters are already meaningful and interesting.

\begin{definition}
    \label{def:obly}
    A deletion-requester $\Y$ is said to be \emph{oblivious} if, when interacting with a data collector $\X$ running $(\pi,\pi_D)$, for any instance of $\pi$ that it initiates, it never accesses the output tape of the corresponding client-side ITI except when running $\pi_D$ to delete this instance, whereupon it merely computes the deletion token and provides it as input to $\pi_D$.
\end{definition}

Note that the deletion-requester $\Y$ not accessing the output tapes does not necessarily mean that the entities or users that it notionally represents similarly do not look at the responses they receive from the data collector -- as long as each user in a collection of users does not communicate anything about such responses to another user, the collection may be faithfully represented by an oblivious $\Y$. Similarly, an oblivious $\Y$ could also represent a single user who sends multiple messages to the data collector, under the condition that the content of these messages, and whether and when the user sends them, does not depend on any information it receives from the data collector.

We also quantify the error that is incurred by a data collector in its deletion-compliance as follows. In our definition of deletion-compliance (\cref{def:compliance}), we required this error to be negligible in the security parameter.

\begin{definition}[Deletion-Compliance Error]
    \label{def:compliance-error}
    Let $k\in\Nat$. Given a data-collector $(\X,\pi,\pi_D)$, an environment $\Z$ and a deletion-requester $\Y$, let $(\state_\X^{R,\secp},\view_\Z^{R,\secp})$ denote the corresponding parts of $\er(\secp)$, and $(\state_\X^{I,\secp},\view_\Z^{I,\secp})$ the corresponding parts of $\ei(\secp)$. The \emph{(statistical) deletion-compliance error} of $(\X, \pi,\pi_D)$ is a function $\eps:\Nat\ra [0,1]$ where for $\secp\in\Nat$, the function value $\eps(\secp)$ is set to be the supremum, over all PPT environments $\Z$, all PPT deletion-requesters $\Y$, and all unbounded distinguishers $D$, of the following quantity when all parties are given $\secp$ as the security parameter:
    \begin{align*}
      \left| \Pr[D(\state_\X^{R,\secp},\view_\Z^{R,\secp})=1] - \Pr[D(\state_\X^{I,\secp},\view_\Z^{I,\secp})=1]\right|
    \end{align*}
    The oblivious deletion-compliance error is defined similarly, but only quantifying over all oblivious PPT deletion-requesters $\Y$. And the $k$-representative deletion-compliance error is defined similarly by quantifying over all $k$-representative PPT $\Y$'s.
\end{definition}

\ifsubmission
\else
\subsubsection{Composition of Deletion-Requesters}
\label{sec:composition-1}
\fi

We show that, for oblivious deletion-requesters, the error in deletion-compliance of any data collector $(\X,\pi,\pi_D)$ grows at most linearly with the number of instances of $\pi$ that are requested to be deleted. In other words, if $k$ different users of $\X$ ask for their information to be deleted, and they all operate independently in the sense that none of them looks at the responses from $\X$ to any of the others, then the error that $\X$ incurs in processing all these requests is at most $k$ times the error it incurs in processing one deletion request.

Apart from being interesting on its own, our reason for proving this theorem is that in the case of some data collectors that we construct in \cref{sec:scenarios}, it turns out to be much simpler to analyze the $1$-representative deletion-compliance error than the error for a generic deletion-requester. The following theorem then lets us go from the $1$-representative error to the error for oblivious deletion-requesters that make more deletion requests.

\begin{theorem}
    \label{thm:comp}
    For any $k\in\Nat$ and any data collector $(\X,\pi,\pi_D)$, the $k$-representative oblivious deletion-compliance error is at most $k$ times its $1$-representative deletion-compliance error.
\end{theorem}

\ifsubmission
\else

\begin{proof}
    We will show this theorem by induction on $k$. Fix some security parameter $\secp$ and suppose the $1$-representative deletion-compliance of a given data collector $(\X,\pi,\pi_D)$ is $\eps_1\in[0,1]$, and its $(k-1)$-representative oblivious deletion-compliance error is $\eps_{k-1}$.

    Consider any $k$-representative oblivious deletion-requester $\Y$ that runs $k$ instances (without loss of generality, exactly $k$) of $\pi$, and then asks for all of them to be deleted. Denote these instances by $\pi_1, \dots, \pi_k$, in the order in which they were initiated by $\Y$, and the corresponding deletion requests by $\pi_{D,1}, \dots, \pi_{D,k}$ (note that it is not necessarily the case that these instances of $\pi_D$ were run in this order). We define a new environment $\Z_1$ and deletion-requester $\Y_1$ based on $\Z$ and $\Y$ as follows:
    \begin{itemize}
      \item $\Z_1$ simulates both $\Z$ and $\Y$, initiating all protocol instances they would, including $\pi_1, \dots, \pi_{k-1}$. When its simulation of $\Y$ is about to initiate $\pi_k$, it instead sends over the simulated state of $\Y$ at that point to $\Y'$. After this, whenever the simulation of $\Y$ attempts to activate the client-side ITI of ``$\pi_k$'', it activates $\Y'$ instead. Otherwise, it proceeds with the rest of its simulation of $\Z$ and $\Y$, again excepting only when $\Y$ tries to initiate $\pi_{D,k}$, at which point it again activates $\Y'$ and does so again whenever the simulation tries to activate the client-side ITI of ``$\pi_{D,k}$''. It then simulates $\Z$ and $\Y$ to completion, and runs the instances $\pi_{D,i}$ for $i$ for which these have not been run by the simulation of $\Y$ yet (also outsourcing $\pi_{D,k}$ to $\Y'$ as earlier if it happens at this stage). 
      \item $\Y_1$, when it is first activated, initiates the instance $\pi_k$ just as $\Y$ would have, and upon subsequent activations it in turn activates the client-side ITI of $\pi_k$ as $\Y$ would have. After $\pi_k$ is complete, the next time $\Y'$ is activated, it initiates $\pi_{D,k}$ with input the deletion token for $\pi_k$, and continues to activate the client-side ITI for the same whenever it is subsequently activated, just at $\Y$ would. 
    \end{itemize}

    Note that it is possible to separate the actions of $\Z$ and $\Y$ into the above $\Z_1$ and $\Y_1$ only because of the obliviousness of $\Y$, which implies that the simulation of $\Y$ does not look at the output tapes of any of the ITIs it controls in order to decide, for instance, when to initiate any protocol instance or activate a client-side ITI, or even what input to provide to these ITIs unless running a $\pi_{D,i}$, and even then it looks only at the ITI of the corresponding $\pi_i$. 
    
    Denote by $\Y_\sim$ the partial simulation of $\Y$ that is run by $\Z_1$, and by $\Z_\sim$ its simulation of $\Z$. Let $(\state_\X^1,\view_{\Z_\sim}^1)$ denote the state of $\X$ and the part of the view of $\Z_1$ that corresponds to the messages sent and received by $\Z_\sim$ at the end of $\textsc{EXEC}^{\X,\pi,\pi_D}_{\Z_1,\Y_1}$. Since $\Y$ is oblivious, the combination of $\Y_\sim$ and $\Y_1$ behaves identically to $\Y$, and we have the following claim. As always, $(\state_\X^R,\view_\Z^R)$ is from the real execution $\er$.

    \begin{claim}
        \label{claim:comp-1}
        $(\state_\X^R,\view_{\Z}^R)$ is distributed identically to $(\state_\X^1,\view_{\Z_\sim}^1)$.
    \end{claim}
    
    Note that $\Y_1$ is a $1$-representative deletion-requester, as it only runs $\pi_k$. Let $(\state_\X^{1,I},\view_{\Z_\sim}^{1,I})$ denote the state of $\X$ and the part of the view of $\Z_1$ that corresponds to the messages sent to and received by $\Z_\sim$ at the end of the corresponding ideal execution $\textsc{EXEC}^{\X,\pi,\pi_D}_{\Z_1,\Y_0}$. By the $1$-representative deletion-compliance of $(\X,\pi,\pi_D)$, we have the following.
    
    \begin{claim}
        \label{claim:comp-2}
        The statistical distance between $(\state_\X^1,\view_{\Z_\sim}^1)$ and $(\state_\X^{1,I},\view_{\Z_\sim}^{1,I})$ is at most $\eps_1$.
    \end{claim}

    Let $\Y_2$ denote the combination of $\Y_\sim$ and $\Y_0$ -- note that this is well-defined since $\Y_\sim$ was isolated from the rest of the simulation in $\Z_1$, and also note that $\Y_2$ is a $(k-1)$-representative oblivious deletion-requester. Let $(\state_\X^2,\view_{\Z}^2)$ denote the state of $\X$ and the view of $\Z$ from the execution $\textsc{EXEC}^{\X,\pi,\pi_D}_{\Z,\Y_2}$. And let $(\state_\X^I,\view_{\Z}^I)$ denote the same from the ideal execution $\textsc{EXEC}^{\X,\pi,\pi_D}_{\Z,\Y_0}$. By the $(k-1)$-representative oblivious deletion-compliance of $(\X,\pi,\pi_D)$, we have the following.

    \begin{claim}
        \label{claim:comp-3}
        The statistical distance between $(\state_\X^2,\view_{\Z}^2)$ and $(\state_\X^I,\view_{\Z}^I)$ is at most $\eps_{k-1}$.
    \end{claim}

    Thus, by the triangle inequality and the above three claims, the statistical distance between $(\state_\X^R,\view_\Z^R)$ and $(\state_\X^I,\state_\Z^I)$ is at most $(\eps_1+\eps_{k-1})$ which, if $\eps_{k-1} \leq (k-1)\eps_1$, is at most $k\eps_1$. By induction, applied to all security parameters $\secp$, this proves the theorem. We finish by proving the above claims.

    \begin{proofof}{\cref{claim:comp-1}}
        This claim follows by first observing that the view of $\X$ is the same in both the execution involving $(\Z,\Y)$ and the one involving $(\Z_1,\Y_1)$, since the latter is a perfect simulation of the former. Thus, the state of $\X$ is the same in both as well. Further, since all messages that are received by the simulation $\Z_\sim$ are exactly those that $\X$ would have sent to $\Z$, its $\view_{\Z_\sim}$ is the same as $\view_\Z$ conditioned on $\state_\X$. This proves the claim.
    \end{proofof}
    
    \begin{proofof}{\cref{claim:comp-2}}
        $\Y_1$ is a deletion-requester that initiates just one instance of $\pi$, and asks for it to be deleted, and we treat $\textsc{EXEC}^{\X,\pi,\pi_D}_{\Z_1,\Y_1}$ as the real execution and $\textsc{EXEC}^{\X,\pi,\pi_D}_{\Z_1,\Y_0}$ as the ideal execution. $1$-representative deletion-compliance now tells us that the statistical distance between $(\state_\X^1,\view_{\Z_1}^1)$ and $(\state_\X^{1,I},\view_{\Z_1}^{1,I})$ is at most $\eps_1$, and the claim follows by observing that the view of $\Z_\sim$ is a subset of the view of $\Z'$.
    \end{proofof}
    
    \begin{proofof}{\cref{claim:comp-3}}
        $\Y_2$ is a deletion-requester that initiates $(k-1)$ instances of $\pi$ and asks for them all to be deleted. Taking $\textsc{EXEC}^{\X,\pi,\pi_D}_{\Z,\Y_2}$ to be the real execution and $\textsc{EXEC}^{\X,\pi,\pi_D}_{\Z,\Y_0}$ to be the ideal execution immediately gives us the claim.
    \end{proofof}
    
\end{proof}

\fi

\ifsubmission
\else
\subsubsection{Composite Data Collectors}
\label{prop:composition-1}
\fi

\ We also show that, given two data collectors that are each deletion-compliant, their combination is also deletion-compliant, assuming obliviousness of deletion-requesters. To be more precise, given a pair of data collectors $(\X_1,\pi_1,\pi_{1,D})$ and $(\X_2,\pi_2,\pi_{2,D})$, consider the ``composite'' data collector $((\X_1,\X_2),(\pi_1,\pi_2),(\pi_{1,D},\pi_{2,D}))$ that works as follows:
\begin{itemize}
  \item An instance of $(\pi_1,\pi_2)$ is either an instance of $\pi_1$ or of $\pi_2$. Similarly, an instance of $(\pi_{1,D},\pi_{2,D})$ is either an instance of $\pi_{1,D}$ or of $\pi_{2,D}$.
  \item The collector $(\X_1,\X_2)$ consists of a simulation of $\X_1$ and of $\X_2$, each running independently of the other.
  \item When processing an instance of $\pi_1$ or $\pi_{1,D}$, it forwards the messages to and from its simulation of $\X_1$, and similarly $\X_2$ for $\pi_2$ or $\pi_{2,D}$.
  \item The state of $(\X_1,\X_2)$ consists of the states of its simulations of $\X_1$ and $\X_2$.
\end{itemize}

Such an $\X$ would represent, for instance, two data collectors that operate separately but deal with the same set of users. We show that, if the constituting data collectors are deletion-compliant, then under the condition of the deletion-requester being oblivious, the composite data collector is also deletion-compliant. 


\begin{theorem}
    \label{thm:comp-collectors}
    If $(\X_1,\pi_1,\pi_{1,D})$ and $(\X_2,\pi_2,\pi_{2,D})$ are both statistically deletion-compliant, then the composite data collector $((\X_1,\X_2),(\pi_1,\pi_2),(\pi_{1,D},\pi_{2,D}))$ is statistically deletion-compliant for oblivious deletion-requesters.
\end{theorem}

The above theorem extends to the composition of any $k$ data collectors in this manner, where there is a loss of a factor of $k$ in the oblivious deletion-compliance error (this will be evident from the proof below).

\begin{proofof}{\cref{thm:comp-collectors}}
    The theorem follows by first showing that the composite collector is deletion-compliant for $1$-representative data collecors, and then applying \cref{thm:comp}. Any $1$-representative deletion-requester $\Y$ interacts either only with (the simulation of) $\X_1$ or with $\X_2$. And since both of these are deletion-compliant, the state of $(\X_1,\X_2)$ and the view of the environment are similarly distributed in both real and ideal executions. Thus, $((\X_1,\X_2),(\pi_1,\pi_2),(\pi_{1,D},\pi_{2,D}))$ is $1$-representative deletion-compliant. Applying \cref{thm:comp} now gives us the theorem.
\end{proofof}

\section{Scenarios}
\label{sec:scenarios}

In this section, we present examples of data collectors that satisfy our definitions of deletion-compliance with a view to illustrate both the modelling of collectors in our framework, and the aspects of the design of such collectors that are necessitated by the requirement of such deletion-compliance.

\subsection{Data Storage and History-Independence}
\label{sec:hist-ind}

Consider the following ostensibly simple version of data storage. A company wishes to provide the following functionality to its users. A user can ask the company to store a single piece of data, say their date-of-birth or a password. At a later point, the user can ask the company to retrieve this data, whence the company sends this stored data back to the user. And finally, the user can ask for this data to be deleted, at which point the company deletes any data the user has asked to be stored.

While a simple task, it is still not trivial to implement the deletion here correctly. The natural way to implement these functionalities is to use a dictionary data structure that stores key-value pairs and supports insertion, deletion and lookup operations. The collector could then store the data a user sends as the value and use a key that is somehow tied to the user, say the user's name or some other identifier. Unless care is taken, however, such data structures could prove insufficient -- data that has been deleted could still leave a trace in the memory implementing the data structure. A pathological example is a dictionary that, to indicate that a certain key-value pair has been deleted, simply appends the string ``deleted'' to the value -- note that such a dictionary can still provide valid insertion, deletion and lookup. While actual implementations of dictionaries do not explicitly maintain ``deleted'' data in this manner, no special care is usually taken to ensure that information about such data does not persist, for instance, in the metadata.


The simplest solution to this problem is to use an implementation of such a data structure that explicitly ensures that the above issue does not occur. History independent data structures, introduced by Micciancio~\cite{Mic97}, are implementations of data structures that are such that their representation in memory at any point in time reveals only the ``content'' of the data structure at that point, and not the history of the operations (insertion, deletion, etc.) performed that resulted in this content. In particular, this implies that an insertion of some data into such a data structure followed by a deletion of the same data would essentially have the same effect on memory as not having done either in the first place.

More formally, these are described as follows by Naor and Teague~\cite{NT01}. Any abstract data structure supports a set of operations, each of which, without loss of generality, returns a result (which may be null). Two sequences of operations $S_1$ and $S_2$ are said to produce the same \emph{content} if for any sequence $T$, the results returned by $T$ with the prefix $S_1$ is the same as the results with the prefix $S_2$. An implementation of a data structure takes descriptions of operations and returns the corresponding results, storing what it needs to in its memory. Naor and Teague then define history independence as a property of how this memory is managed by the implementation.

\begin{definition}
    \label{def:hist-ind}
    An implementation of a data structure is \emph{history independent} if any two sequences of operations that produce the same content also induce the same distribution on the memory representation under the implementation.
\end{definition}

If data is stored by the data collector in a history independent data structure that supports deletion, then being deletion-compliant becomes a lot simpler, as the property of history independence helps satisfy much of the requirements. In our case, we will make us of a history-independent dictionary, a data structure defined as follows. History-independent dictionaries were studied and constructed by Naor and Teague~\cite{NT01}.

\begin{definition}
    \label{def:dict}
    A dictionary is a data structure that stores key-value pairs $(key,value)$, and supports the following operations: 
    \begin{itemize}[topsep=0.4em,itemsep=0em]
      \item $\alg{Insert}(key, value)$: stores the value $value$ under the key $key$. If the key is already in use, does nothing.
      \item $\alg{Lookup}(key)$: returns the value previously stored under the key $key$. If there is no such key, returns $\bot$.
      \item $\alg{Delete}(key)$: deletes the key-value pair stored under the key $key$. If there is no such key, does nothing.
    \end{itemize}
\end{definition}

Our current approach, then, is to implement the data storage using a history-independent dictionary as follows. When a user sends a $(key,value)$ pair to be stored, we insert it into the dictionary. When a user asks for the $value$ stored under a key $key$, we look it up in the dictionary and return it. When a user asks to delete whatever is stored under the key $key$, we delete this from the dictionary. And the deletion, due to history-independence, would remove all traces of anything that was deleted.

There is, however, still an issue that arises from the fact that the channels in our model are not authenticated. Without authentication, any entity that knows a user's $key$ could use it to learn from the data collector whether this user has any data stored with it. And later if the user asks for deletion, the data might be deleted from the memory of the collector, but the other entity has already learnt it, which it could not have done in an ideal execution. In order to deal with this, the data collector has to implement some form of authentication; and further, this authentication, as seen by the above example, has to use some randomness (or perhaps pseudorandomness) generated on the data collector's side. We implement the simplest form of authentication that suffices for this, and the resulting data collector $\histind$ is described informally as follows.

\begin{datacol}
    The data collector $\histind$ maintains a history-independent dictionary $\alg{Dict}$. Below, any information that is not required explicitly to be stored is erased as soon as each message is processed. It waits to receive a message from a user that is parsed as $(\inst,auth,key,value)$, where either of $auth$ or $value$ could be $\bot$, and processed as follows:
\begin{itemize}
  \item If $\inst = \ins{insert}$,
    \begin{itemize}
      \item it samples a new random authentication string $auth$.
      \item it runs $\alg{Dict}.\alg{Insert}((key,auth), value)$ to add $value$ to the dictionary under the key $(key,auth)$.
      \item it responds to the message with the string $auth$.
    \end{itemize}
  \item If $\inst = \ins{lookup}$,
    \begin{itemize}
      \item it recovers the $value$ stored under the key $(key,auth)$ by running $\alg{Dict}.\alg{Lookup}((key,auth))$, and responds with $value$ (if the key is not in use, $value$ will be $\bot$).
    \end{itemize}
  \item If $\inst = \ins{delete}$,
    \begin{itemize}
      \item it deletes any entry under the key $(key,auth)$ by running $\alg{Dict}.\alg{Delete}((key,auth))$.
    \end{itemize}
\end{itemize}

\end{datacol}

The formal description of the above data collector in our framework, along with the associated protocols $\pi$ and $\pi_D$, is presented in \cref{fig:hist-ind-model}. We show that this collector is indeed statistically deletion-compliant.

\ifsubmission
\else
\begin{figure}[h!]\small
    \begin{mdframed}
        \vspace{0.5em}
        
        We model the collector's behavior by a tuple $(\histind,\pi,\pi_D)$. The ITM $\histind$ maintains a history-independent dictionary $\alg{Dict}$ and, given security parameter $\secp$, acts as follows upon activation. Any information below that is not required explicitly to be stored is erased as soon as each message is processed (and just before possibly responding to it and halting or being deactivated).
        \begin{itemize}
          \item It checks whether there is an ITI that it controls whose output tape it has not read so far. If not, it halts.
          \item Otherwise, let $M$ be the first ITI in some arbitrary order whose output has not been read so far. $\histind$ reads this output.
          \item If this output is of the form $(\mathsf{insert}, key, value)$,
            \begin{itemize}
              \item It samples a string $auth \gets \bset^\secp$ uniformly at random.
              \item It runs $\alg{Dict}.\alg{Insert}((key,auth), value)$.
              \item It writes $auth$ onto the input tape of $M$.
            \end{itemize}
          \item If this output is of the form $(\mathsf{lookup}, key, auth)$,
            \begin{itemize}
              \item It runs $\alg{Dict}.\alg{Lookup}((key,auth))$ to get $value$ and writes $value$ on the input tape of $M$.
            \end{itemize}
          \item If this output is of the form $(\mathsf{delete}, key, auth)$,
            \begin{itemize}
              \item It runs $\alg{Dict}.\alg{Delete}((key,auth))$.
            \end{itemize}
        \end{itemize}

        \vspace{0.5em}

        The protocol $\pi$ runs between two parties, the server and the client, and proceeds as follows:
        \begin{itemize}
          \item The client takes as input a tuple $(\inst,key,auth,value)$, where either of $auth$ or $value$ may be empty.
          \item The client sends its input as a message to the server.
          \item The server verifies that $\inst$ is either $\mathsf{insert}$ or $\mathsf{lookup}$.
          \item If $\inst = \mathsf{insert}$ and $value$ is not empty, or if ($\inst = \mathsf{lookup}$ and $auth$ is not empty),
            \begin{itemize}
              \item The server writes the message $(\mathsf{insert},key,value)$ onto its output tape.
              \item The next time it is activated, the server sends the contents of its input tape as a message to the client. We refer to the content of this message as $auth'$.
              \item The client writes the message it receives onto its output tape.
              \item The deletion token for this instance is set to be $(key,auth')$.
            \end{itemize}
          \item If $\inst = \mathsf{lookup}$ and $auth$ is not empty,
            \begin{itemize}
              \item The server writes the message $(\mathsf{lookup},key,auth)$ onto its output tape.
              \item The next time it is activated, the server sends the contents of its input tape as a message to the client.
              \item The client writes the message it receives onto its output tape.
              \item The deletion token for this instance is set to be $\bot$.
            \end{itemize}
        \end{itemize}

        The protocol $\pi_D$ runs between two parties, the server and the client, and proceeds as follows:
        \begin{itemize}
          \item The client takes as input a tuple $(key,auth)$.
          \item The client sends a message $(\mathsf{delete},key,auth)$ to the server.
          \item The server writes the received message onto its output tape.
        \end{itemize}
        
        \vspace{0.5em}
    \end{mdframed}

    \vspace{0.5em}
    \caption{Data Storage Using History-Independence}
    \label{fig:hist-ind-model}
\end{figure}
\fi

\begin{restatable}{theorem}{thmhistind}
    \label{thm:hist-ind}
    The data collector $(\histind,\pi,\pi_D)$ in \cref{fig:hist-ind-model} is statistically deletion-compliant.
\end{restatable}

We present the proof of \cref{thm:hist-ind} in \ifsubmission\else\cref{sec:proofs-hist-ind}\fi. The approach is to first observe that, due to the authentication mechanism, the probability that the environment $\Z$ will ever see any data that was stored by the deletion-requester $\Y$ is negligible in the security parameter. If this never happens, then the view of $\Z$ in the real and ideal executions (where $\Y$ does not store anything) is identical. And when the view is identical, the sequence of operations performed by $\Z$ in the two executions are also identical. Thus, since whatever $\Y$ asks to store it also asks to delete, the state of $\X$ at the end of the execution, due to its use of a history-independent dictionary, depends only on the operations of $\Z$, which are now the same in the real and ideal executions.

In summary, the lessons we learn from this process of constructing a deletion-compliant data collector for data storage are as follow:
\begin{enumerate}
  \item Attention has to be paid to the implementation of the data structures used, which needs to satisfy some notion of independence from deleted data.
  \item Authentication that involves some form of hardness or randomness from the data collector's side has to be employed even to support simple operations.
\end{enumerate}

\subsubsection{Outsourcing Data Storage}
\label{sec:hist-ind-2}

Next, we present a data collector that outsources its storage to an external system, maintaining only bookkeeping information in its own memory. As it actively reveals users' data to this external system, such a data collector cannot be deletion-compliant. However, we show that history-independence can be used to make it \emph{conditionally} deletion-compliant. Again, it turns out to be crucial to ensure that an authentication mechanism is used, for reasons similar to that for the previously constructed data collector. This data collector $\histind_2$ is informally described as follows, and is quite similar to $\histind$.


\begin{datacol}
    The data collector $\histind_2$ maintains a history-independent dictionary $\alg{Dict}$, and interacts with another collector $\W$ that uses the same syntax for messages as the collector $\histind$ from earlier in this section. It waits to receive a message that is parsed as $(\inst,auth,key,value)$, where either of $auth$ or $value$ could be $\bot$, and processed as follows:
    \begin{itemize}
      \item If $\inst = \ins{insert}$,
        \begin{itemize}
          \item It samples a new authentication string $auth$ and a new ``external key'' $exkey$ at random.
          \item It sends the message $(\ins{insert},exkey,value)$ to $\W$ and waits to receive a response $exauth$.
          \item It runs $\alg{Dict}.\alg{Insert}((key,auth), (exkey,exauth))$ to add $(exkey,exauth)$ to the dictionary under the key $(key,auth)$.
          \item It responds to the initial message with the string $auth$.
        \end{itemize}
      \item If $\inst = \ins{lookup}$,
        \begin{itemize}
          \item It recovers the $(exkey,exauth)$ stored under the key $(key,auth)$ by running $\alg{Dict}.\alg{Lookup}((key,auth))$. If the lookup fails, it responds with $\bot$.
          \item It sends the message $(\ins{lookup},exkey,exauth)$ to $\W$ and waits to receive a response $value$.
          \item It responds to the initial message with $value$.
        \end{itemize}
      \item If $\inst = \ins{delete}$,
        \begin{itemize}
          \item It recovers the $(exkey,exauth)$ stored under the key $(key,auth)$ by running $\alg{Dict}.\alg{Lookup}((key,auth))$. If the lookup fails, it halts.
          \item If not, it sends the message $(\ins{delete},exkey,exauth)$ to $\W$.
          \item It deletes any entry under the key $(key,auth)$ by running $\alg{Dict}.\alg{Delete}((key,auth))$.
        \end{itemize}
    \end{itemize}
\end{datacol}        
    
The formal description of the above data collector in our framework, along with the associated protocols $\pi$ and $\pi_D$, is presented in \cref{fig:hist-ind-2-model}. We show that this collector is conditionally deletion-compliant.

\ifsubmission
\else
\begin{figure}[h!]\small
    \begin{mdframed}
        \vspace{0.5em}
        
        We model the collector's behavior by a tuple $(\histind_2,\pi,\pi_D)$, where $\pi$ and $\pi_D$ are as described in \cref{fig:hist-ind-model}. The environment is denoted by $\Z$, and supports $\histind_2$ instantiating instances of protocols $\pi$ and $\pi_D$ with it, with $\histind_2$ on the client side. The ITM $\histind$ maintains a history-independent dictionary $\alg{Dict}$ and, given security parameter $\secp$, acts as follows upon activation. Any information below that is not required explicitly to be stored is erased as soon as each message is processed (and just before possibly responding to it and halting or being deactivated).
        \begin{itemize}
          \item It checks whether there is an ITI that it controls whose output tape it has not read so far. If not, it halts.
          \item Otherwise, let $M$ be the first ITI in some arbitrary order whose output has not been read so far. $\histind$ reads this output.
          \item If this output is of the form $(\mathsf{insert}, key, value)$,
            \begin{itemize}
              \item It samples strings $auth, exkey \gets \bset^\secp$ uniformly at random.
              \item It starts a new session of $\pi$ with $\Z$, and on the input tape of the associated client-side ITI $M'$, it writes $(\ins{insert}, exkey, \bot, value)$.
              \item The next time it is activated, it reads $exauth$ off the output tape of $M'$.
              \item It runs $\alg{Dict}.\alg{Insert}((key,auth), (exkey,exauth))$.
              \item It writes $auth$ onto the input tape of $M$.
            \end{itemize}
          \item If this output is of the form $(\mathsf{lookup}, key, auth)$,
            \begin{itemize}
              \item It runs $\alg{Dict}.\alg{Lookup}((key,auth))$ to get $(exkey,exauth)$. If this lookup fails, it writes $\bot$ on the input tape of $M$ and halts.
              \item Otherwise, it starts a new session of $\pi$ with $\Z$, and on the input tape of the associated client-side ITI $M'$, it writes $(\ins{lookup}, exkey, exauth, \bot)$.
              \item The next time it is activated, it reads $value$ off the output tape of $M'$.
              \item It writes $value$ on the input tape of $M$.
            \end{itemize}
          \item If this output is of the form $(\mathsf{delete}, key, auth)$,
            \begin{itemize}
              \item It runs $\alg{Dict}.\alg{Lookup}((key,auth))$ to get $(exkey,exauth)$. If this lookup fails, it halts.
              \item Otherwise, it starts a new session of $\pi$ with $\Z$, and on the input tape of the associated client-side ITI $M'$, it writes $(\ins{delete}, exkey, exauth, \bot)$.
              \item It runs $\alg{Dict}.\alg{Delete}((key,auth))$.
            \end{itemize}
        \end{itemize}

        \vspace{0.5em}
    \end{mdframed}

    \vspace{0.5em}
    \caption{Data Collector outsourcing its storage while using History-Independent Data Structures}
    \label{fig:hist-ind-2-model}
\end{figure}
\fi

\begin{restatable}{theorem}{thmhistindtwo}
    \label{thm:hist-ind-2}
    Following terminology from \cref{fig:hist-ind-2-model}, the data collector $(\histind_2,\pi,\pi_D)$ is conditionally statistically deletion-compliant in the presence of $(\pi,\pi_D)$.
\end{restatable}

The proof of this theorem is presented in \ifsubmission\else\cref{sec:proofs-hist-ind}\fi. The approach is again to first condition on $\Z$ not being able to guess any of the authentication strings given to $\Y$, an event that happens with overwhelming probability. After this, we show that the history-independence of the dictionary used by $\X$ can be used to effectively split $\X$ into two parts -- one that handles protocols with $\Y$, and the other than handles protocols with $\Z$ -- without affecting what essentially happens in the execution. At this point, we switch to looking at the execution as an auxiliary execution with $\Z$ as the data collector, the first part of $\X$ as the deletion-requester, and the second part as the environment, and apply the auxiliary deletion-compliance of $\Z$ to show that the states of $\Z$ and $\X$ are unchanged if $\Y$ is replaced with a silent $\Y_0$.


\subsection{Data Summarization and Differentially Private Algorithms}
\label{sec:dp}

Here, we demonstrate another approach to constructing deletion-compliant data collectors that perform certain tasks. Consider the case of research organization that wishes to compute compile data about a population and then compute certain statistics on it. It receives such data from a number of volunteers and, once it has enough data, compiles a ``summary'' of the data collected. Later, if some volunteer asks for its data to be deleted, in order to be deletion-compliant, the organization would, in general, also have to modify the summary in order to exclude this volunteer's data which, depending on the summarization procedure used, might not be possible to do without recomputing the summary on the remaining data. We observe, however, that if the computation of the summary satisfies certain privacy properties to begin with, then in fact deletion-compliance can be achieved without altering the summary at all, as long as not too much of the data is requested to be deleted.

The notion of private computation we use in this respect is differential privacy~\cite{DMNS06}, which is defined as follows.

\begin{definition}[Differential Privacy~\cite{DMNS06}]
    \label{def:dp}
    Let $\eps:\Nat\ra [0,1]$ be a function. An algorithm $\alg{A}$ that, for $n \in \Nat$, takes as input $n$-tuples $x = (x_1, \dots, x_n)$ from some domain is said to be $\eps$-differentially private if, for all $n$, for any two $n$-tuples $x$ and $x'$ that differ in at most one location, and any set $S$ in the output space of $\alg{A}$, we have:
    \begin{align*}
      \pr{\algA(x)\in S} \leq e^{\eps(n)} \cdot \pr{\algA(x')\in S}
    \end{align*}
\end{definition}

The idea behind differential privacy is that the output of a private algorithm does not reveal whether any specific value in the tuple was present in the input or not. While there are crucial reasons for using the above condition on the ratios of probabilities in the definition of differential privacy rather than a bound on the statistical distance, for our purposes the following implication is sufficient.

\begin{fact}
    \label{fact:dp}
    If an algorithm $\algA$ is $\eps$-differentially private, then, for any two $n$-tuples $x$ and $x'$ that differ in at most one location, the statistical distance between the distributions of $\algA(x)$ and $\algA(x')$ is at most $\eps(n)$.
\end{fact}

We use the above guarantee to design a data collector for the aforementioned task where the summary computed is more-or-less the same in both the real and ideal executions if the deletion-requester in the real execution enters and deletes at most one data point. The central idea is to simply compute the summary in a differentially private manner (rather, the data collector we construct is deletion-compliant if the computation of the summary is differentially private).

However, this turns out to not be sufficient for a couple of reasons. The first is again the issue of authentication. While, unlike the data storage example earlier, here the data collector does not provide a lookup mechanism that could leak one user's data to another, lack of authentication would enable other attacks. For instance, if there were no authentication, then $\Z$ could send to $\X$ some data that it knows $\Y$ will try to delete later, and this deletion would happen only in the real execution. The same authentication mechanism as before (here represented by the random choice of the key that is sampled by the collector) handles such issues.

The second reason is that the point at which the summary is computed has to be decided carefully. For instance, suppose there is some $n$ such that the collector always computes the summary once it receives $n$ data points. Then, if $\Z$ enters exactly $(n-1)$ data points and $\Y$ enter a single point, then the summary would be computed in the real execution but not in the ideal. For this reason, we need to randomize the point at which the summary is computed, by picking $n$ at random. However, care is to be taken to always compute the summary a subset of points of the same size always, as a differentially private mechanism is allowed to leak the number of points in its input.

Accounting for all this, we design the data collector described informally below and formally within our framework in \cref{fig:dp-model}.

\begin{datacol}
    The data collector $\diffp$ maintains a history-independent dictionary $\alg{Dict}$. Given security parameter $\secp$, it first samples an integer $thr$ (the threshold) uniformly at random from $[\secp+1,2\secp]$, initiates a counter $count$ to $0$, and a boolean flag $summarized$ to $\mathsf{false}$. It waits to receive a message from a user that is either a data point $x$ or a deletion request, and processes it as follows: 
    \begin{itemize}
      \item If it receives a data point $x$,
        \begin{itemize}
          \item it samples a new key $key$ uniformly at random from $\bset^\secp$.
          \item it runs $\alg{Dict}.\alg{Insert}(key, x)$ to add $x$ to the dictionary under the key $key$.
          \item if a new entry was actually inserted by the above operation, it increments $count$ by $1$, and if, further, $summarized$ is $\mathsf{true}$, it also increments $thr$.
          \item if $count = thr$ and $summarized = \mathsf{false}$,
            \begin{itemize}
              \item pick a random subset $S$ of the values in $\alg{Dict}$ such that $|S| = \secp$. 
              \item compute and store $summary$ as the output of $\alg{summarize}(S)$.
              \item set $summarized$ to $\mathsf{true}$.
            \end{itemize}
          \item it responds to the message with the string $key$.
        \end{itemize}
      \item If it receives a deletion request for a key $key$, it deletes any entry under the key $key$ by running $\alg{Dict}.\alg{Delete}(key)$. If an entry is actually deleted, it decrements $count$; and if, further, $summarized$ is $\mathsf{true}$, it also decrements $thr$.
    \end{itemize}
\end{datacol}

\ifsubmission
\else
\begin{figure}[h!]\small
    \begin{mdframed}
        \vspace{0.5em}
        
        We model the collector's behavior by a tuple $(\diffp,\pi,\pi_D)$ that uses an algorithm $\alg{summarize}$ that takes a dictionary as input and outputs a string. The ITM $\diffp$, given security parameter $\secp$, is initially set up as follows:
        \begin{itemize}
          \item It initializes a history-independent dictionary $\alg{Dict}$.
          \item It samples $thr$ uniformly at random from $[\secp+1,2\secp]$.
          \item It sets $count \gets 0$, $summarized \gets \mathsf{false}$, and $summary \gets \bot$.
        \end{itemize}
        Upon activation, it acts as follows. Any information below that is not required explicitly to be stored is erased as soon as each message is processed (and just before possibly responding to it and halting or being deactivated).
        \begin{itemize}
          \item It checks whether there is an ITI that it controls whose output tape it has not read so far. If not, it halts.
          \item Otherwise, let $M$ be the first ITI in some arbitrary order whose output has not been read so far. $\diffp$ reads this output.
          \item If this output is of the form $(\mathsf{insert}, value)$,
            \begin{itemize}
              \item It samples a string $key \gets \bset^\secp$ uniformly at random.
              \item It runs $\alg{Dict}.\alg{Insert}(key, value)$.
              \item If the above operation actually results in an insertion, increment $count$. If, additionally, $summarized = \mathsf{true}$, increment $thr$.
              \item If $count = thr$ and $summarized = \mathsf{false}$,
                \begin{itemize}
                  \item Sample a random subset $S$ of the values stored in $\alg{Dict}$ such that $|S| = \secp$.
                  \item Set $summary \gets \alg{summarize}(S)$.
                  \item Set $summarized \gets \mathsf{true}$.
                \end{itemize}
              \item It writes $key$ onto the input tape of $M$.
            \end{itemize}
          \item If this output is of the form $(\mathsf{delete}, key)$,
            \begin{itemize}
              \item It runs $\alg{Dict}.\alg{Delete}(key)$.
              \item If this operation actually resulted in a deletion,
                \begin{itemize}
                  \item Decrement $count$.
                  \item If $summarized = \mathsf{true}$, decrement $thr$.
                \end{itemize}
            \end{itemize}
        \end{itemize}

        \vspace{0.5em}

        The protocol $\pi$ runs between two parties, the server and the client, and proceeds as follows:
        \begin{itemize}
          \item The client takes as input a tuple $value$.
          \item The client sends its input as a message to the server.
          \item The server writes the message $(\mathsf{insert},value)$ onto its output tape.
          \item The next time it is activated, the server sends the contents of its input tape as a message to the client. We refer to the content of this message as $key$.
          \item The client writes the message it receives onto its output tape.
          \item The deletion token for this instance is set to be $key$.
        \end{itemize}
        
        The protocol $\pi_D$ runs between two parties, the server and the client, and proceeds as follows:
        \begin{itemize}
          \item The client takes as input a string $key$.
          \item The client sends its input as a message to the server.
          \item The server writes the message $(\mathsf{delete},key)$ onto its output tape.
        \end{itemize}
        
        \vspace{0.5em}
    \end{mdframed}

    \vspace{0.5em}
    \caption{Data Collector employing Differentially Private Summarization}
    \label{fig:dp-model}
\end{figure}
\fi

We show that, as long as there is at most one deletion request, the error in deletion-compliance of the above data collector is not much more than the error in the privacy of the summarization procedure. Note that, following \cref{thm:comp} regarding the composition of deletion requests, it may be inferred that for oblivious deletion-requesters making $k$ deletion requests (which includes the significant case of $k$ volunteers asking for their data to be deleted independently of each other), the deletion-compliance error is at most $k$ times the error in the statement below.

\begin{restatable}{theorem}{thmdp}
    \label{thm:dp}
    Suppose the mechanism $\alg{summarize}$ used by the data collector $\diffp$ in \cref{fig:dp-model} is $\eps$-differentially private. Then, $(\diffp,\pi,\pi_D)$ has $1$-representative statistical deletion-compliance error at most $(\eps+1/\secp+\poly(\secp)/2^{\secp})$. 
\end{restatable}

We prove \cref{thm:dp} in \ifsubmission\else\cref{sec:proofs-dp}\fi. The ideas behind it are: \begin{enumerate*} [label=(\itshape\roman*\upshape)] \item due to the authentication mechanism, $\Z$ cannot interfere with the operations initiated by $\Y$; \item since the dictionary is history-independent, arguments from \cref{sec:hist-ind} may be used to show that it is the same in the real and ideal executions; \item since the point of computation of the summary is chosen at random, the probability that it is computed in the real execution but not in the ideal one is low; and \item since the summary is computed in a differentially private way, and the contents of the dictionary differ by at most one entry in the real and ideal executions, the distributions of the summary are similar. \end{enumerate*}

The lessons learnt from our construction here are:
\begin{enumerate}
  \item Usage of algorithms satisfying certain notions of privacy, such as differential privacy, could enable deletion-compliance without requiring additional action on the part of the data collector.
  \item Certain aspects of the data collector may need to be hidden or have some entropy, such as the exact point at which the summary is computed by $\diffp$. In this case, this essentially ensured that the real and ideal executions were treated similarly by the collector.
  \item Authentication is necessary even if there is no lookup mechanism where the collector explicitly reveals stored data.
\end{enumerate}


\subsection{Deletion in Machine Learning}
\label{sec:ml}

Finally, we note that existing notions and algorithms for data deletion in machine learning can be used to construct deletion-compliant data collectors that run perform such learning using the data they collect. Recently, Ginart et al~\cite{GGVZ19} defined a notion of data deletion for machine learning algorithms that can be used together with history-independent data structures to maintain a learned model while respecting data deletion requests. We first rephrase their definition in our terms as follows for ease of use.

We will be concerned with a learning algorithm $\alg{learn}$ that takes as input a dataset $D = \set{x_i}$ consisting of data entries $x_i$ indexed by $i \in [|D|]$, and outputs a hypothesis $h$ from some hypothesis space $H$. Suppose there is an algorithm $\alg{delete}$ that takes as input a dataset $D$, a hypothesis $h$, and an index $i$, and outputs another hypothesis from $H$. For any $i\in [|D|]$, let $D_{-i}$ denote the dataset obtained by removing the $i^{\text{th}}$ entry.

\begin{definition}[\cite{GGVZ19}]
    \label{def:ml}
    The algorithm $\alg{delete}$ is a \emph{data deletion operation} for $\alg{learn}$ if, for any dataset $D$ and index $i \in [|D|]$, the outputs of $\alg{learn}(D_{-i})$ and $\alg{delete}(D,\alg{learn}(D),i)$ are identically distributed.
\end{definition}

Given that $\alg{delete}$ is also given $D$ in its input, there is always the trivial deletion operation of retraining the model from scratch using $D_{-i}$, but the hope is that in many cases it is possible to delete more efficiently than this. And Ginart et al~\cite{GGVZ19} show that this is indeed possible in certain settings such as $k$-means clustering.

We present below (and formally in \cref{fig:ml-model}) an example of a data collector that makes use of any learning algorithm with such a deletion operation to maintain a learned model while allowing the data used in its training to be deleted. For simplicity, we assume that both the learning and deletion algorithms work with datasets represented implicitly by dictionaries -- the data is represented as the set of $value$'s stored in the dictionary, and the $key$'s of the dictionary are used as a proxy for the index $i$ above. Note that this can be done without loss of generality, as either algorithm could simply start by going through all the $key$'s in the dictionary in some fixed order, and writing down the values as rows of the dataset. The data collector operates very similarly to that from \cref{sec:dp}, but this time there are no concerns about the size of the dataset being revealed, as the deletion operation takes care of this.


\begin{datacol}
    The data collector $\ml$ maintains a dataset as a history-independent dictionary $\alg{Dict}$. Given security parameter $\secp$, it first samples an integer $thr$ uniformly at random from $[\secp,2\secp]$, and sets a boolean flag $learnt$ to $\mathsf{false}$. It waits to receive either a data point or a deletion request from a user, and acts as follows:
    \begin{itemize}
      \item If it receives a data point $x$,
        \begin{itemize}
          \item it samples a new key $key$ uniformly at random from $\bset^\secp$.
          \item if $learnt = \mathsf{false}$,
            \begin{itemize}
              \item it runs $\alg{Dict}.\alg{Insert}(key, x)$ to add $x$ to the dictionary under the key $key$.
              \item if a new entry was actually inserted by the above operation, it increments $count$ by $1$, and if, further, $learnt$ is $\mathsf{true}$, it also increments $thr$.
              \item if $count = thr$,
                \begin{itemize}
                  \item compute and store $model$ as the output of $\alg{learn}(\alg{Dict})$.
                  \item set $learnt$ to $\mathsf{true}$.
                \end{itemize}
            \end{itemize}
          \item it responds to the message with the string $key$.
        \end{itemize}
      \item If it receives a deletion request for the key $key$,
        \begin{itemize}
          \item it first runs $\alg{Dict}.\alg{Lookup}(key)$. If the lookup succeeds,
            \begin{itemize}
              \item it updates $model$ to be the output of $\alg{delete}(\alg{Dict},model,key)$.
              \item it runs $\alg{Dict}.\alg{Delete}(key)$.
              \item it decrements $count$ by 1.
              \item if $learnt = \mathsf{true}$, it decrements $thr$ by $1$.
            \end{itemize}
        \end{itemize}
    \end{itemize}
\end{datacol}       

\ifsubmission
\else
\begin{figure}[h!]\small
    \begin{mdframed}
        \vspace{0.5em}
        
        We model the collector's behavior by a tuple $(\ml,\pi,\pi_D)$ that uses a learning algorithm $\alg{learn}$ that takes a dictionary as input, and a corresponding deletion operation $\alg{delete}$. The ITM $\ml$, given security parameter $\secp$, is initially set up as follows:
        \begin{itemize}
          \item It initializes a history-independent dictionary $\alg{Dict}$.
          \item It samples $thr$ uniformly at random from $[\secp,2\secp]$.
          \item It sets $count \gets 0$, $learnt \gets \mathsf{false}$, and $model \gets \bot$.
        \end{itemize}
        Upon activation, it acts as follows. Any information below that is not required explicitly to be stored is erased as soon as each message is processed (and just before possibly responding to it and halting or being deactivated).
        \begin{itemize}
          \item It checks whether there is an ITI that it controls whose output tape it has not read so far. If not, it halts.
          \item Otherwise, let $M$ be the first ITI in some arbitrary order whose output has not been read so far. $\histind$ reads this output.
          \item If this output is of the form $(\mathsf{insert}, value)$,
            \begin{itemize}
              \item It samples a string $key \gets \bset^\secp$ uniformly at random.
              \item If $learnt = \mathsf{false}$,
                \begin{itemize}
                  \item It runs $\alg{Dict}.\alg{Insert}(key, value)$.
                  \item If the above operation actually results in an insertion, increment $count$. If, additionally, $learnt = \mathsf{true}$, increment $thr$.
                  \item If $count = thr$,
                    \begin{itemize}
                      \item Set $model \gets \alg{learn}(\alg{Dict})$.
                      \item Set $learnt \gets \mathsf{true}$.
                    \end{itemize}
                \end{itemize}
              \item It writes $key$ onto the input tape of $M$.
            \end{itemize}
          \item If this output is of the form $(\mathsf{delete}, key)$,
            \begin{itemize}
              \item It runs $\alg{Dict}.\alg{Lookup}(key)$. If this lookup succeeds,
                \begin{itemize}
                  \item It sets $model \gets \alg{delete}(\alg{Dict},model,key)$.
                  \item It runs $\alg{Dict}.\alg{Delete}(key)$.
                  \item It decrements $count$ by $1$.
                  \item If $learnt = true$, it decrements $thr$ by $1$.
                \end{itemize}
            \end{itemize}
        \end{itemize}

        \vspace{0.5em}

        The protocol $\pi$ runs between two parties, the server and the client, and proceeds as follows:
        \begin{itemize}
          \item The client takes as input $(key,value)$.
          \item The client sends its input as a message to the server.
          \item The server writes the message $(\mathsf{insert},key,value)$ onto its output tape.
          \item The next time it is activated, the server sends the contents of its input tape as a message to the client. We refer to the content of this message as $key$.
          \item The client writes the message it receives onto its output tape.
          \item The deletion token for this instance is set to be $key$.
        \end{itemize}
        
        The protocol $\pi_D$ runs between two parties, the server and the client, and proceeds as follows:
        \begin{itemize}
          \item The client takes as input a tuple $key$.
          \item The client sends its input as a message to the server.
          \item The server writes the message $(\mathsf{delete},key)$ onto its output tape.
        \end{itemize}
        
        \vspace{0.5em}
    \end{mdframed}

    \vspace{0.5em}
    \caption{Data Collector maintaining a model learnt from data that could be deleted}
    \label{fig:ml-model}
\end{figure}
\fi

\begin{restatable}{theorem}{thmml}
    \label{thm:ml}
    The data collector $(\ml,\pi,\pi_D)$ as described in \cref{fig:ml-model} has $1$-representative deletion-compliance error at most $(1/\secp+\poly(\secp)/2^\secp)$.
\end{restatable}

We prove \cref{thm:ml} in \ifsubmission\else\cref{sec:proofs-ml}\fi, along the same lines as \cref{thm:dp}.




\bibliographystyle{alpha}
\bibliography{abbrev1,crypto,refs}

\appendix

\ifsubmission
\else
\section{Proofs for \cref{sec:scenarios}}
\label{sec:proofs}

Throughout these proofs, we will use $\sdist{X}{Y}$ to denote the statistical distance between distributions $X$ and $Y$. We will be using certain properties of statistical distance stated below. 

\begin{fact}
    \label{fact:sd-conditional}
    For any jointly distributed random variables $(X,Y)$ and $(X',Y')$ over the same domains,
    \begin{align*}
      \sdist{(X,Y)}{(X',Y')} \leq \sdist{X}{X'} + \expec{x\gets X}{\sdist{Y_x}{Y'_x}}
    \end{align*}
    where $Y_x$ is the distribution of $Y$ conditioned on $X = x$, and $Y'_x$ is the distribution of $Y'$ conditioned on $X' = x$.
\end{fact}

\begin{proof}
    By the definition of statistical distance, we have:
    \begin{align*}
      \sd((X,Y),(X',Y')) &= \sum_{x,y} \abs{\pr{X=x \wedge Y=y} - \pr{X'=x \wedge Y'=y}} \\
                         &= \sum_{x,y} \abs{\pr{X=x}\cdot \pr{Y=y\ |\ X=x} - \pr{X'=x}\cdot \pr{Y'=y\ |\ X'=x}}\\
                         &= \sum_{x,y} \large|\pr{X=x}\cdot \pr{Y=y\ |\ X=x} - \pr{X=x}\cdot \pr{Y'=y\ |\ X'=x} \\ &\qquad + \pr{X=x}\cdot \pr{Y'=y\ |\ X'=x} - \pr{X'=x}\cdot \pr{Y'=y\ |\ X'=x}\large|\\
                         &\leq \sum_x \pr{X=x} \cdot \sum_y \abs{\pr{Y=y\ |\ X=x} - \pr{Y'=y\ |\ X'=x}} \\ &\qquad + \sum_x \abs{\pr{X=x} - \pr{X'=x}} \cdot \sum_y \pr{Y'=y\ |\ X'=x}\\
                         &= \expec{x\gets X}{\sdist{Y_x}{Y'_x}} + \sdist{X}{X'}
    \end{align*}
\end{proof}

\subsection{Data Storage and History-Independence}
\label{sec:proofs-hist-ind}

In this subsection, we restate and prove the theorems from \cref{sec:hist-ind} about the deletion-compliance of data collectors using history-independent data structures to store users' data.

\subsubsection{Proof of \cref{thm:hist-ind}}

\thmhistind*

The state $\state_{\histind}$ at any point (when it halts or is deactivated) consists only of the memory being used by the implementation of the dictionary $\alg{Dict}$, which we are given is history-independent. Fix any environment $\Z$ and deletion-requester $\Y$, both of which run in time $\poly(\secp)$ given security parameter $\secp$. For some $\secp$, let $(\state_\histind^R,\view_\Z^R)$ be the corresponding parts of $\exec^{\histind,\pi,\pi_D}_{\Z,\Y}(\secp)$ and $(\state_\histind^I,\view_\Z^I)$ the corresponding parts of of $\exec^{\histind,\pi,\pi_D}_{\Z,\Y}(\secp)$.

To start with, note that $\histind$ does not reveal any information about any data stored unless the appropriate authentication string $auth$ is sent with a $\ins{lookup}$ request. As $auth$ is chosen at random each time and $\Y$ does not communicate with $\Z$ (directly or indirectly by requesting deletion, etc., on its behalf), the environment $\Z$ cannot tell whether $\Y$ even exists or not unless it gets lucky and guesses an $auth$ string that it has not seen. This leads us to the following claim that we prove later.

\begin{claim}
    \label{claim:hist-ind-1}
    The statistical distance between $\view_\Z^I$ and $\view_\Z^R$ is at most $\poly(\secp)/2^\secp$.
\end{claim}

Next, we observe that if the insertions and deletions made by $\Z$ are the same, then the state of the dictionary maintained by $\X$ is the same in the real and ideal executions. This is because in the real execution, whatever $\Y$ asks to insert (and only this), it also asks to delete, and the rest follows by the history-independence of the dictionary.

\begin{claim}
    \label{claim:hist-ind-2}
    For any $\view \in \supp(\view_\Z^I)$, the statistical distance between the distribution of $\state_\histind^I$ conditioned on $\view_\Z^I = \view$ and the distribution of $\state_\histind^R$ conditioned on $\view_\Z^R = \view$ is at most $\poly(\secp)/2^\secp$.
\end{claim}

Together, \cref{claim:hist-ind-1,claim:hist-ind-2} imply that the distributions $(\state_\histind^I,\view_\Z^I)$ and $(\state_\histind^R,\view_\Z^R)$ have statistical distance at most $\poly(\secp)/2^\secp$, by \cref{fact:sd-conditional}. We complete our proof by proving the above claims.

\begin{proofof}{\cref{claim:hist-ind-1}}
    Denote by $R_\Z$ the random variable corresponding to the randomness string, if any, used by $\Z$ during its execution. Suppose $\Z$, in the course of the real execution, engages in at most $q$ ``interactions'' $\Psi_1, \dots, \Psi_q$. Each $\Psi_i$ is a random variable that may be an execution of $\pi$ or $\pi_D$ with $\histind$, or a message sent to $\Y$, an activation of $\X$ or $\Y$, or a declaration of the end of the Alive phase. Note that $q$ is at most $\poly(\secp)$ since the running time of $\Z$ is at most this much. The view of $\Z$ is then given by the tuple $(R_\Z,\Psi_1,\dots,\Psi_q)$. Distinguishing between real and ideal executions, we write $\view_\Z^I = (R_\Z^I,\Psi_1^I,\dots,\Psi_q^I)$, and $\view_\Z^R = (R_\Z^R,\Psi_1^R,\dots,\Psi_q^R)$.


    For any $i\in [q]$, we show that conditioned on $(R_\Z^R, \Psi_1^R, \dots, \Psi_{i-1}^R) = (R_\Z^I, \Psi_1^I, \dots, \Psi_{i-1}^I) = (r_\Z, \psi_1, \dots, \psi_{i-1})$ for any $r_\Z$ and $\psi_1, \dots, \psi_{i-1}$ in the appropriate domains, the distribution of $\Psi_i^R$ and $\Psi_i^I$ have statistical distance at most $p(\secp)/2^{\secp}$ for some polynomial $p$. Then, by repeated application of \cref{fact:sd-conditional} (and since $R_\Z^R$ is identical to $\R_\Z^I$), the distance between $\view_\Z^R$ and $\view_\Z^I$ is at most $qp(\secp)/2^\secp = \poly(\secp)/2^{\secp}$.

    Once we fix $(R_\Z^R, \Psi_1^R, \dots, \Psi_{i-1}^R)$ and $(R_\Z^I, \Psi_1^I, \dots, \Psi_{i-1}^I)$ as above, all the possibilities for $\Psi_i^R$ or $\Psi_i^I$ are listed below. Note that, as each $\Psi$ is initiated by $\Z$, the variables $\Psi_i^R$ and $\Psi_i^I$ are identically constituted until $\Z$ \emph{receives} a messages during the execution.
    \begin{enumerate}
      \item $\Psi_i^R$ and $\Psi_i^I$ are both messages to $Y$. In this case, as the state of $\Z$ at this point is the same in both real and ideal executions, the message sent is also the same in both cases, and $\Psi_i^R = \Psi_i^I$.
      \item $\Psi_i^R$ and $\Psi_i^I$ are both activations of $\X$ or of $\Y$, or declarations of the end of the Alive phase. For the same reason as in the last case, $\Psi_i^R = \Psi_i^I$.
      \item $\Psi_i^R$ and $\Psi_i^I$ are both executions of $\pi$ of the $\ins{insert}$ type. Again, the messages sent by $\Z$ in both cases are the same. The responses to the messages is an $auth$ that is sampled by $\histind$ at random independently of anything else in the system. Thus, $\Psi_i^R$ and $\Psi_i^I$ are identically distributed in this case.
      \item $\Psi_i^R$ and $\Psi_i^I$ are both executions of $\pi_D$. Again, the messages sent by $\Z$ in both cases are the same, and there are no responses to it. So $\Psi_i^R = \Psi_i^I$.
      \item $\Psi_i^R$ and $\Psi_i^I$ are both executions of $\pi$ of the $\ins{lookup}$ type. Again, the messages sent by $\Z$ in both cases are the same -- say it is $(\ins{lookup},key,auth)$. The response to such a message is a $value$ that is retrieved by $\histind$ by looking up the key $(key,auth)$ in its dictionary. And unless this key-auth pair corresponds to one that is used by $\Y$ in the real execution of the protocol in an $\ins{insert}$ execution of $\pi$ or in an execution of $\pi_D$, the response $value$ is the same in the real and ideal executions. Further, by our requirement of deletion-requesters, the set of key-auth pairs used in any execution of $\pi_D$ by $\Y$ is a subset of those that occur in any $\ins{insert}$ execution of $\pi$ initiated by it. So the probability that $\Psi_i^R$ and $\Psi_i^I$ disagree is at most the probability that $auth$ was selected as the authentication string during some $\ins{insert}$ request by $\Y$. As $\Y$ makes at most $\poly(\secp)$ such requests and $auth$ is drawn at random by $\histind$ from the space $\bset^\secp$, the probability that this happens is at most $\poly(\secp)/2^\secp$.
    \end{enumerate}
    Thus, in all cases, the statistical distance is at most $\poly(\secp)/2^{\secp}$. This proves the claim. 
\end{proofof}

\begin{proofof}{\cref{claim:hist-ind-2}}
    We will prove this by showing that, once the view of $\Z$ is fixed to some $\view$ in the support of $\view_\Z^I$, the contents of $\dict$ at the end of both the real and ideal executions are the same. The claim then follows by the history-independence of $\dict$ (which implies then that the \emph{state} of the memory implementing $\dict$ is also the same), and the fact that $\state_\histind$ whenever $\histind$ halts or is deactivated is just the memory used to implement $\dict$.

    In the ideal execution, the sequence of operations performed on $\dict$ are exactly those that follow from instances of $\pi$ and $\pi_D$ initiated by $\Z$, as the silent $\Y_0$ does not do anything. So in the end, the contents of $\dict$ are those key-value pairs that were inserted due to an $\ins{insert}$ instance of $\pi$ by $\Z$, but never deleted by a corresponding instance of $\pi_D$, as specified by $\view$. 

    In the real world, the sequence of operations on $\dict$ are those that follow from instances of $\pi$ and $\pi_D$ initiated by $\Z$ and also by $\Y$. However, by design, $\Y$ always runs $\pi_D$ for every instance of $\pi$ that it initiates, and never runs $\pi_D$ for an instance of $\pi$ that it did not initiate. Thus, unless there is a collision in the $key$ and $auth$ in an $\ins{insert}$ instance of $\pi$ initiated by $\Y$ and an instance of $\pi$ or $\pi_D$ initiated by $\Z$, the contents of $\dict$ at the end are again exactly the key-value pairs that were inserted by an $\ins{insert}$ instance of $\pi$ initiated by $\Z$ and never deleted by a corresponding instance of $\pi_D$ run by $\Z$. 

    The probability that there is such a collision in $(key,auth)$ is at most $\poly(\secp)/2^\secp$, as there are at most $\poly(\secp)$ instances of $\pi$, and each $auth$ is selected at random from $\bset^\secp$ during an $\ins{insert}$. (Note that since we are conditioning on a $\view$ in the support of $\view_\Z^I$, any $\ins{lookup}$ request by $\Z$ using a $(key,auth)$ that it did not insert would fail, and this could only decrease the probability of such a collision given a certain number of instances.) Thus, conditioning on the view of $\Z$ being fixed in this manner, except with probability at most $\poly(\secp)/2^\secp$ over the randomness of $\histind$, the contents of the dictionary are distributed identically in the real and ideal worlds, thus proving the claim.
\end{proofof}

\subsubsection{Proof of \cref{thm:hist-ind-2}}

\thmhistindtwo*

Fix any environment $\Z$ such that $(\Z,\pi,\pi_D)$ is statistically deletion-compliant in the presence of $(\pi,\pi_D)$, and any deletion-requester $\Y$. We prove the conditional deletion-compliance of $(\histind_2,\pi,\pi_D)$ by showing that for any such $\Z$ and $\Y$, the states of $\histind_2$ and $\Z$ at the end of execution is almost identical to those of a sequence of different configurations of entities. Fix some security parameter $\secp$. Throughout the real execution, if $\Z$ at some point uses a string $auth$ that was given to $\Y$ during an $\ins{insert}$ execution of $\pi$, the data collector's guarantees fail. But, as $auth$ is chosen at random from $\bset^\secp$, this happens with probability at most $\poly(\secp)/2^{\secp}$. We count this probability towards the error and, in the rest of the proof, we condition on this event not happening. The guarantees would also fail if $\Y$ executed $\pi_D$ to ask for deletion of an instance of $\pi$ it did not initiate, but by definition of $\Y$ this does not happen. Throughout the rest of the proof, we also refer to $\histind_2$ as $\X$ for ease of notation.

Let $(\state_\X^R,\state_\Z^R)$ be the corresponding parts of the result of the real execution $\erpi(\secp)$. We define a new data collector $\X_1$ based on $\X$ and $\Y$ to use as a hybrid in our arguments. $\X_1$ essentially simulates two instance of $\X$ -- $\X_{1,\Z}$,  which handles the instances of $\pi$ initiated by $\Z$, and $\X_{1,\Y}$, which acts independently (using a separate dictionary) and handles instances of $\pi$ initiated by $\Y$. Note that an $\X_1$ is not really a valid data collector as a data collector has no idea which protocols are initiated by $\Y$ and which by $\Z$, but we only consider $\X_1$ in the proof after fixing $\Z$ and $\Y$, where it is well-defined. Let $\state_{\X_{1,\Z}}^1$ represent the part of the state of $\X_1$ that corresponds to the memory used by the simulation $\X_{1,\Z}$, and $\state_{\Z}^1$ the state of $\Z$, both in the execution $\exec^{\X_1,\pi,\pi_D}_{\Z,\Y,(\pi,\pi_D)}$. Note that this execution looks identical to $\Z$ to the real execution (due to the conditioning on no collisions with $\Y$). By the history-independence of the dictionary used, we have the following:

\begin{claim}
    \label{claim:hist-ind-2-1}
    $(\state_\X^R,\state_\Z^R)$ and $(\state_{\X_{1,\Z}}^1,\state_{\Z}^1)$ are distributed identically.
\end{claim}

Next, we reorganize the entities $\X_1$ and $\Y$ into two new entities $\X_2$ and $\Y_2$ in order to fit them in an auxiliary execution that we will set up. $\X_2$ is essentially $\X_{1,\Z}$, but additionally will play the role of the environment in the auxiliary execution. In order to do this, whenever it is activated in the course of the execution by default (that is, not because of a message that was sent to it or an output by one if the ITIs it controls), it simply activates $\Z$, which will act as the data collector in this auxiliary execution. $\Y_2$, on the other hand, acts as the combination of $\Y$ and $\X_{1,\Y}$: it simulates $\Y$ and $\X_{1,\Y}$, and whenever $\Y$ initiates a protocol instance of $\pi$ or $\pi_D$, it simulates its interaction with $\X_{1,\Y}$, and when $\X_{1,\Y}$ tries to initiate a protocol with $\Z$, $\Y_2$ runs this protocol with the actual $\Z$ in its place. And when $\Z$ sends a message intended for $\Y$, $\Y_2$ forwards this to its simulation of $\Y$.

It may be seen that, under the above setup, $\cexec^{\Z,\pi,\pi_D}_{\X_2,\Y_2,(\pi,\pi_D)}(\secp)$ is a valid auxiliary execution. Let $(\state_{\X_2}^2,\state_{\Z}^2)$ be the corresponding parts of this execution. It may be seen that all we have done is reorganize the various entities and the actual computation and storage going on is the same as in the previous execution. Thus, we have the following,

\begin{claim}
    \label{claim:hist-ind-2-2}
    $(\state_{\X_{1,\Z}}^1,\state_{\Z}^1)$ and $(\state_{\X_2}^2,\state_{\Z}^2)$ are distributed identically.
\end{claim}

Let $(\state_{\X_2}^3,\state_{\Z}^3)$ be the corresponding parts of the ideal auxiliary execution $\cexec^{\Z,\pi,\pi_D}_{\X_2,\Y_0,(\pi,\pi_D)}(\secp)$. By the auxiliary-deletion-compliance of $(\Z,\pi,\pi_D)$ in the presence of $(\pi,\pi_D)$, we have the following.

\begin{claim}
    \label{claim:hist-ind-2-3}
    The distributions of $(\state_{\X_2}^2,\state_{\Z}^2)$ and $(\state_{\X_2}^3,\state_{\Z}^3)$ have negligible statistical distance.
\end{claim}

Let $(\state_\X^I,\state_\Z^I)$ be the corresponding parts of the ideal execution $\eipi(\secp)$. In a manner identical to \cref{claim:hist-ind-2-1,claim:hist-ind-2-2}, we can show the following.

\begin{claim}
    \label{claim:hist-ind-2-4}
    The distributions of $(\state_{\X}^I,\state_{\Z}^I)$ and $(\state_{\X_2}^3,\state_{\Z}^3)$ are distributed identically.
\end{claim}

Together, \cref{claim:hist-ind-2-1,claim:hist-ind-2-2,claim:hist-ind-2-3,claim:hist-ind-2-4} show that, conditioned on the $auth$ strings of $\Y$ not being used by $\Z$, the states of $\X$ and $\Z$ in the real and ideal executions have negligible statistical distance. As this event being conditioned on happens except with negligible probability as well, the theorem follows.

\subsection{Data Summarization and Differentially Private Algorithms}
\label{sec:proofs-dp}

\thmdp*

The state $\state_\diffp$ of the data collector $\diffp$ consists of the memory used to implement and store the history-independent dictionary $\dict$, the summary $summary$ ($\bot$ or otherwise), the threshold $thr$, and the variables $count$ and $summarized$.

Fix any environment $\Z$ and a $1$-representative deletion-requester $\Y$. To start with, the only messages that $\Z$ receives from $\diffp$ are the $key$ strings, which are chosen at random independently of everything else, and so we have the claim below from arguments similar to those for \cref{claim:hist-ind-1}.

\begin{claim}
    \label{claim:dp-1}
    $\view_\Z^R$ and $\view_\Z^I$ are distributed identically.
\end{claim}

Fix any view of $\Z$ (to something from an ideal execution). That is, fix the sequence and content of all messages sent and received by $\Z$. We next show that, even under this conditioning, the state of $\diffp$ is distributed almost the same at the end of the real and ideal executions.

\begin{claim}
    \label{claim:dp-2}
    Conditioned on $\view_\Z^R = \view_\Z^I = \view$ for some $\view\in\supp(\view_\Z^I)$, the statistical distance between the distributions of $\state_\diffp^R$ and $\state_\diffp^I$ is at most $\eps + 1/\secp + \poly(\secp)/2^{\secp}$.
\end{claim}

Using \cref{fact:sd-conditional}, the \cref{claim:dp-1,claim:dp-2} together prove \cref{thm:dp}. We now finish by proving \cref{claim:dp-2}.

\begin{proofof}{\cref{claim:dp-2}}
    $\state_\diffp$ consists of the memory used to implement the dictionary $\alg{Dict}$, the variables $thr$, $count$ and $summarized$, and the $summary$ that may have been computed during the execution. A bad case in which $\diffp$ actually performs poorly is when the $key$ given to $\Y$ during an insertion is also used at some point by $\Z$, in an insertion or deletion request. However, as $key$ is chosen uniformly at random from $\bset^\secp$, this happens with probability at most $\poly(\secp)/2^{\secp}$ over the randomness of $\diffp$. We add this probability of failure to our error and, in the rest of the proof, condition on this event not happening. (The other bad event is if $\Y$ asks for deletion of an instance of $\pi$ that it did not execute, but by the requirements on $\Y$ this does not happen.) After the above conditioning, all valid insertion and deletion requests by $\Z$ and $\Y$ result in insertions and deletions as expected. Under this conditioning, we fix the view of $\Y$ to anything valid, and show that the conclusion of \cref{claim:dp-2} holds even for this fixed view of $\Y$.

    Next, we observe that conditioning on any fixing of the views of $\Z$ and $\Y$, the distributions, at the end of the executions, of $\alg{Dict}$ and the tuple $(thr,count,summarized,summary)$ are independent of each other, for the following reason. Fixing the views of $\Z$ and $\Y$ fixes the sequence of $(key,value)$ pairs that are inserted into and deleted from $\alg{Dict}$. Since $\alg{Dict}$ is history-independent, its distribution (that is, the distribution of the memory used to implement it) is determined completely by the content resulting from this sequence of insertions and deletions. Due to how it is computed, the distribution of $summary$ is determined completely by the content of $\alg{Dict}$ at the point when it contain $thr$ entries for the first time; both this and the values of $thr$, $count$ and $summarized$ are determined completely by the sequence of valid insertions and deletions in the execution. Thus, fixing the views makes these two distributions independent.

    Due to this independence, it is sufficient to argue separately that $\alg{Dict}$ and the above tuple are distributed similarly in the two executions. Since we have fixed the view of $\Z$, and hence the sequence of insertions and deletions made by it, and $\Y$ also deletes whatever it inserts, the contents of $\alg{Dict}$ are indeed the same at the end of the real and ideal executions. Next we argue about the distributions of $thr$, $count$ and $summary$ (and ignore $summarized$ since, without loss of generality, it is completely determined by $summary$).

    Let $vars_n^R$ and $vars_n^I$ denote the distributions of $(thr,count,summary)$ in the real and ideal executions, respectively, when $thr$ is chosen by $\diffp$ to be $n$. The distribution of this tuple in the real and ideal executions are equal convex combinations of the $vars_n^R$'s and $vars_n^I$'s, respectively, for the various values of $n \in [\secp+1,2\secp]$. Below we essentially show that for most values of $n$, there is an $n'$ (and vice versa) such that the distribution $vars_n^R$ is almost identical to $vars_{n'}^I$.
    
    Suppose there are a total of $m$ valid insertions in the sequence of operations initiated by $\Z$ (which we have fixed according to the conditioning in the claim being proven). Together with the insertion by $\Y$, there are a total of $(m+1)$ insertions that we denote by $\ins{ins}_1,\dots,\ins{ins}_{m+1}$. Of these, the $i$ corresponding to the insertion by $\Y$ we denote by $i_{\Y}$, and the $i$ that corresponds to the insertion that just precedes the deletion request by $\Y$ by $i_{\Y,D}$.

    We define two functions $f_R, f_I: [\secp+1,2\secp] \ra ([m+1]\cup\set{\bot})$ that are defined as follows. For any $n \in [\secp+1,2\secp]$, $f_R(n)$ is set to be the $i$ such that when $\diffp$ sets $thr=n$ in the beginning, $summary$ in the real execution is computed at $\ins{ins}_i$, or $\bot$ if $summary$ is never computed in this case. Similarly, $f_I(n)$ is set to be the $i$ such that $summary$ is computed in the ideal execution at the insertion corresponding to $\ins{ins}_i$ with $thr=n$ in the beginning, and $\bot$ if it is never computed. Observe that \begin{enumerate*} [label=(\itshape\roman*\upshape)] \item $summary$ is computed the first time the number of entries in $\alg{Dict}$ equals $thr$, and \item the number of entries in $\alg{Dict}$ just after $\ins{ins}_i$ is the same in the ideal and real executions if $i < i_\Y$ or $i > i_{\Y,D}$, and otherwise is exactly one more in the real execution than in the ideal. \end{enumerate*} These two observations immediately give the following claim. 

    \begin{claim}
        \label{claim:dp-3}
        The functions $f_R$ and $f_I$ have the following properties for any $n\in[\secp+1,2\secp]$:
        \begin{enumerate}
          \item If $f_R(n) < i_\Y$, then $f_I(n) = f_R(n)$.
          \item If $f_R(n) > i_{\Y,D}$, then $f_I(n) = f_R(n)$.
          \item If $f_R(n) = \bot$, then $f_I(n) = \bot$.
          \item If $n$ is such that $f_R(n) \in [i_\Y,i_{\Y,D}]$ and it is not the least $n$ with this property, then $f_R(n) = f_I(n-1)$.
        \end{enumerate}
    \end{claim}

    \begin{proofof}{\cref{claim:dp-3}}
        The first three statements follow immediately from the observations listed above the statement of the claim, and we prove the fourth now. Suppose $n$ is such that both $f_R(n)$ and $f_R(n-1)$ are in $[i_\Y,i_{\Y,D}]$. This implies that $f_R(n)$ is the $i$ such that the number of entries in $\alg{Dict}$ equals $n$ for the first time in the real execution. By observation (ii) above, the number of entries in $\alg{Dict}$ in the ideal execution at this point is $(n-1)$, and we claim that this is also the first time this happens (that is, $f_I(n-1) = f_R(n)$). If not, then either $f_I(n-1) < i_\Y$, or $f_I(n-1) \in [i_Y,f_R(n))$. In the former case, $f_R(n-1) = f_I(n-1)$, but we know this cannot happen as $f_R(n-1) \geq i_\Y$. In the latter case, there would be $n$ entries in $\alg{Dict}$ in the real execution at the point $f_I(n-1)$, implying that $f_R(n)$ is less than itself. Thus, we reach a contradiction in both cases, proving the claim.
    \end{proofof}
    
    Suppose the initial value of $thr$ was set to $n$. Note that $summary$ is computed using the contents of $\alg{Dict}$ in the real execution just after $\ins{ins}_{f_R(n)}$, and $summary$ in the ideal execution just after $\ins{ins}_{f_I(n)}$ (if these function values are not $\bot$). We look at the following exhaustive list of cases for $n$:
    \begin{enumerate}
      \item Either $f_R(n) < i_Y$, or $f_R(n) > i_{\Y,D}$. In these cases, by \cref{claim:dp-3}, $f_I(n) = f_R(n)$. And in these cases, when $f_R(n) \neq \bot$, the contents of $\alg{Dict}$ in the real execution after $\ins{ins}_{f_R(n)}$ is the same as the contents in the ideal execution after $\ins{ins}_{f_I(n)}$. Thus, in all these cases, the $summary$ computed in the real and ideal executions are distributed identically. Further, at the end of both the real and ideal executions, $thr$ and $count$ are set to be the number of entries still remaining in $\alg{Dict}$, which is again the same in both real and ideal executions. Thus, in this case, $vars_n^R$ and $vars_n^I$ are distributed identically.
      \item $f_R(n) = \bot$. By \cref{claim:dp-3}, $f_I(n)$ is also $\bot$. In both the real and ideal executions, $summary$ is never computed, and at the end $thr$ is still $n$ and $count$ is the number of entries in $\alg{Dict}$, which is the same in both executions. Thus, in this case too, $vars_n^R$ and $vars_n^I$ are distributed identically.
      \item $f_R(n) \in [i_\Y,i_{\Y,D}]$, and $n$ is not the least with this property. In this case, by \cref{claim:dp-3}, $f_R(n) = f_I(n-1) = i$, say. The contents of $\alg{Dict}$ after $\ins{ins}_i$ in the real and ideal executions differ by at most one entry -- that inserted by $\Y$. The $\eps$-differential privacy of $\alg{summarize}$ now guarantees that the distributions of $summary$ computed in the real and ideal executions have statistical distance at most $\eps$. Further, at the end of the execution, both $thr$ and $count$ are equal to the number of entries in $\alg{Dict}$, which is the same in both executions. Thus, the statistical distance between $vars_n^R$ and $vars_{n-1}^I$ is at most $\eps$.
      \item $f_R(n) \in [i_\Y,i_{\Y,D}]$, and $n$ is the least with this property. In this case there is no distribution in the ideal case that we can compare the distribution of $summary_n^R$ to, but there is just one such value of $n$.
    \end{enumerate}
    Thus, except for one value of $n$ that is chosen with probability at most $1/\secp$, each $vars_n^R$ is at most $\eps$-far from some unique $vars_{n'}^I$. Thus, the statistical distance between the distributions of $(thr,count,summary)$ in the real and ideal executions is at most $(\eps+1/\secp)$, conditioned on the bad event of a repeated $key$ not occuring. As this event happens with probability at most $\poly(\secp)/2^\secp$ as noted at the beginning of the proof, the statistical distance between the distribution of $(thr,count,summary)$ is at most $(\eps+1/\secp+\poly(\secp)/2^\secp)$, proving the claim.

\end{proofof}

\subsection{Deletion in Machine Learning}
\label{sec:proofs-ml}

\thmml*

Fix any environment $\Z$ and a deletion-requester $\Y$ that runs at most one execution of $\pi$. As the environment $\Z$ only receives randomly chosen $key$'s as messages, its view in both the real and ideal executions are identical.

\begin{claim}
    \label{claim:ml-1}
    $\view_\Z^R$ and $\view_\Z^I$ are distributed identically.
\end{claim}

Similar to the proof of \cref{thm:dp}, we show that, conditioned on the view of $\Z$ being fixed to any specific view in both the real and ideal executions, the distribution of the state of $\X$ is similar. That is, we show the following claim that, together with the above observations about the views and \cref{fact:sd-conditional}, proves \cref{thm:ml}.

\begin{claim}
    \label{claim:ml-2}
    Conditioned on $\view_\Z^R = \view_\Z^I = \view$ for some $\view\in\supp(\view_\Z^I)$, the statistical distance between the distributions of $\state_\ml^R$ and $\state_\ml^I$ is at most $(1/\secp + \poly(\secp)/2^{\secp})$.
\end{claim}

\begin{proofsketchof}{\cref{claim:ml-2}}
    The proof of \cref{claim:ml-2} follows along the same lines as that of \cref{claim:dp-2}, being identical all the way up to \cref{claim:dp-3} except using $model$ instead of $summary$. In the case analysis following the statement of \cref{claim:dp-3} there, the first, second and fourth cases hold here as is, and in the third case, though the $model$ computed initially is different, due to the perfect deletion operation of the algorithm $\alg{learn}$, the distribution of $model$ at the end of the execution is the same in the real and ideal executions. Thus, the error this time is only from the fourth case, and is $1/\secp$. Together with the error from the possible repetition of $\Y$'s $key$, this proves the claim.
\end{proofsketchof}

\fi

\end{document}